\documentclass[opre,nonblindrev, hidelinks]{informs4}
\usepackage{setspace} % put this package above hyperref, then the footnote hyperlink works well
\usepackage{amsfonts}
\usepackage{amssymb}
\usepackage{graphicx}
\newcommand{\citeay}[1]{\citeauthor{#1} \citeyear{#1}}
\usepackage{natbib}
\usepackage{hyperref}
\usepackage{subfigure}
\usepackage{multicol}
\usepackage{multirow}
\usepackage{epstopdf}
\usepackage{xypic}
\usepackage{float}
\usepackage{dsfont}

\usepackage{textcomp}
\usepackage{listings}
\usepackage{bbm}
\usepackage{bm}

\usepackage{eurosym}
\usepackage{enumerate}
\usepackage{setspace}
\usepackage{float}
\usepackage{mathrsfs}
\usepackage{setspace}
\usepackage{comment}

\RequirePackage{endnotes}
\OneAndAHalfSpacedXI

\usepackage{url}
\makeatletter
\g@addto@macro{\UrlBreaks}{\UrlOrds}
\makeatother

\usepackage{natbib}

\bibpunct[, ]{(}{)}{,}{a}{}{,}%
\def\bibfont{\small}%
\newcommand{\OUT}[1]{}
\TheoremsNumberedThrough

\begin{document}

\TITLE{An Implementation Relaxation Approach to Principal-Agent Problems}
\ARTICLEAUTHORS{	
\AUTHOR{Hang Jiang}
\AFF{National University of Singapore, School of Computing, Department of Information Systems and Analytics,  \EMAIL{jianghang@u.nus.edu}}		
%\AUTHOR{Chen Jin}
%\AFF{National University of Singapore, School of Computing, Department of Information Systems and Analytics,\\ \EMAIL{disjinc@nus.edu.sg} }
%\AUTHOR{Luyi Yang}
%\AFF{University of California, Berkeley, Haas School of Business, \EMAIL{luyiyang@haas.berkeley.edu} }
}

\ABSTRACT{
The classic first-order approach (FOA) relaxes the principal-agent problem by replacing the incentive compatibility (IC) constraint with its first-order condition. We show that FOA is not a valid relaxation when the support of the outcome distribution shifts with the agent’s effort, as in well-studied additive-noise models. In such cases, the optimal effort may occur at a kink point that the first-order condition cannot capture, causing FOA to miss optimal contracts, including the widely adopted bonus schemes. Motivated by this limitation, we introduce the Implementation Relaxation Approach (IRA), which accommodates non-differentiable optima and is straightforward to apply across settings. Rather than directly relaxing IC, IRA relaxes the set of implementable agent efforts and utilities, reducing the problem to identifying the effort-utility pair from which the optimal contract can be constructed. This inverse perspective is particularly convenient for analyzing simple contracts. Using IRA, we derive an optimality condition for quota-bonus contracts that is more general than FOA-based conditions, including those established in the literature under fixed-support assumptions. This also fills a gap where the optimality of quota-bonus contracts in shifting-support settings has been examined only under endogenous assumptions, and it highlights the broader applicability of IRA as a methodological tool.}

\maketitle	
\RUNAUTHOR{}
\RUNTITLE{}

\section{Introduction}\label{intro}
Moral hazard principal-agent models have been studied for decades and form a cornerstone of incentive and contract design under information asymmetry and outcome uncertainty. At their core, these models take the form of a bilevel optimization problem: a principal seeks to maximize expected utility by choosing an optimal contract subject to two standard constraints---incentive compatibility (IC), which ensures that the agent selects an effort level that maximizes utility given the contract, and individual rationality (IR), which guarantees that the agent’s expected utility is no less than an exogenously specified reservation utility.

In the canonical framework with a continuum of efforts \citep{holmstrom1979moral}, the first-order approach (FOA) is commonly employed: it replaces the intractable IC constraint with its first-order condition, thereby yielding a relaxed problem. Similar to other relaxation methods in optimization \citep{gorry1972relaxation}, the FOA is designed to enlarge the constraint set over which the objective is maximized, preserving all feasible contracts in the original problem \citep{rogerson1985first}. Extensive exogenous conditions have been proposed to guarantee that the agent’s expected utility as a function of effort under the optimal contract is (quasi-) concave, strictly unimodal, or decreasing-then-increasing, thereby allowing the analysis to focus solely on the relaxed problem without losing optimal contracts \citep{rogerson1985first, jewitt1988justifying, conlon2009two,  jung2015information, chade2020no, jiang2024revisiting}.

% with $h$ increasing in effort $a$ and $\eta$ denoting non-negative random noise
However, these validations of the FOA rely invariably on a fixed‐support assumption: the support of the outcome distribution remains unchanged with effort. This raises the question of how the FOA performs once this fixed-support dependence is relaxed. In contrast, several influential studies \citep[e.g.,][]{oyer2000theory, dai2016impact} focus on additive-noise models, where outcomes take the form of effort plus a non-negative random noise term.\footnote{\citet{oyer2000theory} has inspired a substantial line of research in Marketing and Operations Management on sales-force compensation, for example, \cite{chu2013salesforce} and \cite{dai2013salesforce,dai2016impact,dai2019salesforce}. Beyond sales-force contexts, shifting support can also arise in other practices. For example, when a secretary exerts effort to screen out weak candidates, the resulting distribution of applicants is left-truncated.} This specification implies a class of \textit{shifting‐support settings} where higher effort raises the lower bound of outcomes. Yet these studies invoke the FOA and restrict attention to the relaxed problem, without specifying the conditions under which it preserves the optimality of the simple contracts (e.g., quota-bonus) under study. This leaves open whether the FOA remains a reliable relaxation in shifting‐support settings, and whether strictly unimodal agent utility continues to suffice for validating it, as in fixed-support settings.

%In both \citet{oyer2000theory} and \citet{dai2016impact}, the outcome is modeled as $x= h(e) + \eta$, where $h$ is an increasing function and $\eta$ is a non-negative random noise. 
This paper revisits a canonical principal-agent problem where both the principal and agent are risk
neutral, and the agent has limited liability.
In Section~\ref{IFOAV}, we show that in shifting‐support settings, the FOA is not a valid relaxation. We construct a counterexample within the additive-noise framework in which a simple quota-bonus contract is optimal (thus feasible) for the original problem but infeasible (and suboptimal) in the FOA‐based relaxed problem, despite the agent’s expected utility being strictly unimodal. The failure arises because, in shifting‐support settings, the agent’s expected utility under the optimal contract may not be everywhere differentiable, despite all other assumptions coinciding with those in the fixed-support models studied in the literature, as in \citet{holmstrom1979moral}. In our counterexample, the agent's optimal effort level lies at a kink point, a non‐differentiable interior optimum that cannot be captured by the first-order condition, and thus by the FOA. (By contrast, prior studies have documented FOA failures arising either from boundary solutions---for example, Example 2 in \citeay{ke2018monotonicity}---or from non-optimal stationary points that occur when the agent’s utility is non-unimodal, as in \citeay{jiang2024revisiting}.) Moreover, in this example, the FOA erroneously suggests that quota-bonus contracts are strictly dominated by linear contracts, whereas in fact the reverse holds. This divergence further underscores that reliance on the FOA can mislead commission-versus-bonus comparisons, a common focus in the sales‐force compensation literature \citep[e.g.,][]{kishore2013bonuses, chung2014bonuses, schottner2017optimal}.

%---and suggest that similar failures may arise for other bonus structures, such as two‐tier bonuses \citep{wang2021moral}they may not be well suited since
These limitations substantially diminish the FOA’s practical relevance---particularly given the prevalence and simplicity of quota-bonus contracts \citep{oyer2000theory}. Alternative approaches for principal-agent problems have been proposed in the literature; however, in addition to their reliance on fixed-support assumptions, as noted by \citet{jiang2024revisiting}, they either focus on strictly risk-averse agents \citep[e.g.,][]{ke2018monotonicity}, sacrifice optimality guarantees \citep{wang2021moral}, or restrict analysis to discrete cases \citep{grossmanhart, dai2021incentive}. The added complication of shifting support may further limit their applicability. This highlights the need for new methods capable of addressing the challenges posed by non-differentiability.

In Section~\ref{IRA}, we introduce the Implementation Relaxation Approach (IRA)---a genuine relaxation that accommodates kink optima (as well as stationary optima), preserves the full feasible set of the original problem, and retains both generality and applicability. Instead of purporting to directly relax the IC constraint, the IRA relaxes the \textit{implementable set},\footnote{The term ``implementable set" also appears in \citet{jiang2024revisiting}, who focus on implementation in the relaxed problem. To the best of my knowledge, the implementation idea can be traced back to in \citet{grossmanhart}, who focuses solely on implementing the agent’s effort, whereas here we emphasize implementation of both effort and utility and propose a relaxation.} the set of all implementable effort-utility pairs (combinations of agent effort and agent utility). Since this set is endogenous and difficult to characterize, the IRA approximates it with a tractable, exogenously defined superset. We then identify the optimal effort-utility pair from this superset that maximizes the principal's expected payoff; if this pair also belongs to the original implementable set, the true optimal contracts are subsequently recovered from this pair.

An immediate application of the IRA is to establish optimality conditions for simple contracts, as verifying whether a contract of a given parametric form implements a specific effort-utility pair reduces to a straightforward parameter verification problem. In Section~\ref{OCQBC}, we employ the IRA to derive the exogenous optimality condition for quota-bonus contracts. Our condition applies to both fixed‐ and shifting‐support settings, thereby filling another gap left by \citet{oyer2000theory}, who established the optimality of quota-bonus contracts in shifting-support settings via the FOA but offered no exogenous conditions to substantiate it.

%namely the absence of exogenous optimality conditions for quota-bonus contracts in shifting‐support settings.

Moreover, our condition guarantees the existence of an optimal contract, whereas FOA-based conditions merely impose local shape restrictions on the agent’s utility to justify applying FOA and do not address existence \citep{jiang2024revisiting}. Beyond this, we show our exogenous condition is even more general than the endogenous FOA-validity condition in characterizing quota-bonus contract optimality, as it encompasses all cases in which a quota-bonus contract is optimal for both the original problem and its FOA-based relaxed problem. This further implies that our condition subsumes all FOA-based optimality conditions for quota-bonus contracts, including those developed under fixed-support settings \citep{park1995incentive, kim1997limited, jiang2024revisiting}. Taken together, these results highlight the IRA's broader applicability as a methodological tool.

In sum, our contributions are as follows. First, we identify a precise relaxation failure of the FOA caused by non-differentiability—its exclusion of kink optima and the resulting loss of optimal (quota-bonus) contracts in shifting-support settings. Second, we introduce a new and genuine relaxation approach, the Implementation Relaxation Approach (IRA), which accommodates non-differentiable scenarios and can be readily applied across settings to derive optimality conditions for simple contracts. Third, we apply the IRA to obtain optimality conditions for quota-bonus contracts that are valid in both fixed- and shifting-support settings, thereby filling a gap left by the absence of exogenous conditions in the shifting-support case. Our condition not only guarantees the existence of optimal contracts, but is also strictly more general than any FOA-based optimality conditions for quota-bonus contracts.

\section{The Model}\label{model}
A principal (she) engages an agent (he) to execute a task. The agent selects an effort level $a \geq 0$ to influence the outcome $X(a)$, a continuous random variable with support $[L(a),\, \bar{x}]$, where $L(a)$ is the lower bound of the support and $\bar{x} > 0$ is the exogenously given upper bound, which may be $+\infty$.
%The signal $x$ on the outcome received by the principal is an element of the sample space $X$ of a probability space. Let $\mu(\cdot \mid a)$ be a parametrized probability measure on the set of signals $X$. There are both discrete and continuous formulations in the literature. To cover both cases, we assume $X$  a non-negative random variable with a cumulative distribution function $F(x\mid a)$. that $F(x\mid a) = \int \limits_{-\infty}^{x} f(y\mid a)dy$, with the “mass‐density’’ functionThe distribution of $X(a)$ may be discrete, continuous, or a mixture of both, depending on the context.\mathrm{supp}\bigl(X(a)\bigr) = 
The function $L : [0, \infty) \to [0, \bar{x})$ is nondecreasing on $[0, \infty)$ and continuously differentiable on $(0, \infty)$. Effort $a$ guarantees the realized outcome at least $L(a)$. The agent’s effort $a$ incurs disutility $c(a)$, where function $c\colon[0,\infty)\to[0,\infty)$ is twice differentiable and satisfies $c'(a)>0$, $c''(a)\ge 0$ for all $a \geq 0$. The outcome \( X(a) \) has a probability density function \( f(x \mid a) \) and a cumulative distribution function $F(x\mid a)$. Moreover, for each fixed $x\in(L(a),\bar x]$, both $f(x\mid a)$ and $F(x\mid a)$ are differentiable in $a>0$, with derivatives $F_a(x\mid a)=\partial_aF(x\mid a)$ and $f_a=\partial_af(x\mid a)$. We assume both $F_a$ and $f_a$ are continuous on $(L(a),\bar x] \times (0, \infty)$, and $F_a<0$ on this domain,  indicating that higher effort leads to a first-order stochastic improvement in the output.

All model primitives are common knowledge, except that the agent’s effort 
$a$ is privately observed by the agent. Consequently, the principal, being unable to observe $a$, designs a contract based solely on the realized outcome \(x\). The contract satisfies \textit{limited liability} and is specified by a function $s: \mathbb{R}_{\geq 0} \to \mathbb{R}_{\geq 0}$, where $s(x) \geq 0$ for all $x \geq 0$. Both the principal and the agent are risk-neutral, and the agent has a (exogenously given) \textit{reservation utility} $u_0 \geq 0$.
 
%The principal is risk-neutral, while the agent is risk-neutral or strictly risk-averse. The agent's utility over wealth \(w\) is represented by \(u(w)\), where \(u(\cdot)\) is an increasing and weakly concave function with \(u(0) = 0\). 
%Both the principal and the agent are utility maximizers.\footnote{The utility function is fundamentally an ordinal function, designed to reflect the preference ordering between different levels of wealth, rather than depending on the absolute values of utility. Consequently, normalization (e.g., setting \(u(0) = 0\)) or any constant shift does not alter the agent's preferences or risk attitudes over wealth, as such transformations preserve the relative differences in utility. For further discussion, see \citet{varian1992microeconomic}.}
%Following \cite{holmstrom1979moral}, we assume the agent's total utility is additively separable and thus is expressed as \(u(w) - c(a)\).% where \(w\) represents the monetary transfer from the principal. 

Under contract $s$ and effort level $a$, the principal’s expected utility is
\begin{equation}\label{principalpayoff}
    E^P(a, s) \;=\; \mathbb{E}[\,X(a) \mid a \,] - \mathbb{E}[\,s(X(a)) \mid a \,] = \; \int_{L(a)}^{\bar x} \bigl[x - s(x)\bigr]\,f(x \mid a)\,dx ,
\end{equation}
and the agent’s expected utility is
\begin{equation}\label{agentpayoff}
    E^A(a, s) \;=\; \mathbb{E}[\,s(X(a)) \mid a \,] - c(a) = \int_{L(a)}^{\bar x} s(x)\,f(x \mid a)\,dx \;-\; c(a).
\end{equation}

The game proceeds as follows: first, the principal offers a contract $s(x)$ that satisfies the Limited Liability (LL) constraint, i.e., $s(x) \geq 0, \, \forall x\geq 0$. The agent then selects an effort level that maximizes his expected utility, known as the Incentive Compatibility (IC) constraint. The agent accepts the contract if and only if another Individual Rationality (IR) constraint is satisfied: the agent's expected utility under the chosen effort is at least  \(u_0\). Upon accepting the contract, the agent receives compensation \(s(x)\) after the outcome is realized, and the principal retains a payoff of \(x - s(x)\).

The principal chooses an optimal contract that satisfies IC, IR, and LL, as formulated below:
\[
\begin{alignedat}{3}
\textbf{(OP)}\quad & \max_{ a_s,\; s} && E^P( a_s,s),\\[2pt]
\text{subject to}\quad 
& \text{(IC)} \ \ \; &&  a_s \in \arg\max_{a\ge 0} E^A(a,s),\\
& \text{(IR)} \ \ \; && E^A( a_s,s) \ge u_0,\\
& \text{(LL)}\ \ \; && s(x) \ge 0, \ \ \forall x\geq 0.
\end{alignedat}
\]

Here, \textbf{OP} denotes the principal’s original problem.
Note that without the limited liability constraint, the principal could simply sell the business to the agent or impose severe penalties whenever \(x < L(a)\) to ensure that the agent’s effort is at least \(a\). Such contracts are unrealistic, as agents typically lack deep pockets and contracts often include liability caps \citep{sappington1983limited}. Thus, the limited liability constraint is essential for maintaining both the realism and nontriviality of the model. Moreover, our model encompasses both fixed‐ and shifting‐support settings. Formally,
\begin{definition}\label{shiftingsupport}
We say the principal's problem is in the \emph{fixed‐support setting} if \(L(a)\) is constant over \(a > 0\). By contrast, it is in the \emph{shifting‐support setting} if \(L(a) > L(a')\) for some \(a > a'> 0\).
\end{definition}

We assume that when the agent has multiple indifferent choices, he chooses the one most favorable to the principal. We also impose the following mild regularity assumption.
\begin{assumption}\label{ass:interchange} 
The following conditions hold:
\begin{itemize}
    \item[(i)] For all \(a \ge 0\), \(\mathbb{E}[X(a)\mid a] < \infty\), and the mapping 
\(a \mapsto \mathbb{E}[X(a)\mid a]\) is continuous on \([0,\infty)\).
    \item[(ii)] \(\lim_{a \to \infty} \left\{\mathbb{E}[X(a) \mid a] - c(a)\right\} = -\infty\).
    \item[(iii)] The mapping \(a \mapsto \sup_{x\in(L(a),\,\bar x]} \frac{f_a(x\mid a)}{f(x\mid a)} \text{ is continuous on } (0,\infty)\).
    \item[(iv)] For any contract \(s\) with \(\mathbb{E}[s(X(a)) \mid a] < \infty\), the following interchange is valid:
    \[
    \lim_{h \uparrow 0}
    \int_{L(a)}^{\bar x}
    \frac{f(x\mid a+h)-f(x\mid a)}{h}\,s(x)\,\mathrm{d}x
    \;=\;
    \int_{L(a)}^{\bar x}
    \lim_{h \uparrow 0} \frac{f(x\mid a+h)-f(x\mid a)}{h}\,s(x)\,\mathrm{d}x.
    \]
\end{itemize}
\end{assumption}

In assumption \ref{ass:interchange}, parts (i) and (iii) ensure that the expected outcome and the supremum likelihood ratio vary continuously with effort, and part (ii) ensures that the optimal effort is finite \citep{jiang2024revisiting}. Part (iv) justifies the interchange of the limit and the integral. Specifically, literature using the first-order approach writes \(\frac{\partial}{\partial a}\!\int s(x) f(x\mid a)\,dx = \int s(x)\, f_a(x\mid a)\,dx\)  \citep{holmstrom1979moral, jewitt1988justifying}.
This step implicitly assumes part (iv)—namely, that the limit in the difference quotient may be interchanged with the integral. A standard justification for the interchange is provided by the Dominated Convergence Theorem. For example, a simple sufficient condition is that, for each fixed $a$, the likelihood ratio is bounded, i.e., $|\frac{f_a(x\mid a)}{f(x\mid a)}|\le M$ for a.e.\ $x$ within the support (a local version of the uniformly bounded likelihood-ratio assumption used by \citeay{kadan2013minimum}).\footnote{Bounded likelihood ratio is sufficient for assumption \ref{ass:interchange}-(iv). Given that $|\frac{f_a(x\mid a)}{f(x\mid a)}|\le M$ for a.e.\ $x$ within the support, for all sufficiently small $|h|$,
\(\big|\frac{f(x\mid a+h)-f(x\mid a)}{h}\big|\le 2M\,f(x\mid a)\ \text{a.e.}\), so the integrand is bounded by
\(
\big|\,s(x)\,\frac{f(x\mid a+h)-f(x\mid a)}{h}\,\big|
\le 2M\,s(x)\,f(x\mid a).
\)
Since $\mathbb{E}[\,s(X(a))\,]=\int s(x)\,f(x\mid a)\,dx<\infty$, dominated convergence applies and yields
\(
\lim_{h\uparrow 0}\int s(x)\,\frac{f(x\mid a+h)-f(x\mid a)}{h}\,dx
= \int s(x)\,f_a(x\mid a)\,dx 
\), where we extend $f(\cdot\mid a)$ by zero outside its support.}

%This requires that for a given $a$, there exists an integrable function $g(x)$ that serves as an envelope for the difference quotient in a neighborhood of $a$:$\sup_{\{0<|h|<\delta\}} \left| \frac{f(x\mid a+h) - f(x\mid a)}{h} \right| \le g(x) \ \text{for almost every } x.$Together with $\mathbb{E}[\,|s(X(a))|\,] < \infty$, the theorem directly justifies the result: $\lim_{h\to 0} \int s(x) \frac{f(x\mid a+h)-f(x\mid a)}{h}\,\mathrm{d}x = \int s(x) f_a(x\mid a)\,\mathrm{d}x.$

\begin{comment}
     \footnote{%
Fix \(a>0\) and suppose there exist \(\delta>0\) and an \(x\)-integrable function \(g\) such that 
\(\lvert f_a(x\mid t)\rvert\le g(x)\) for all \(x\) and all \(t\in[a-\delta,a+\delta]\).  
Then for \(|h|<\delta\) the mean–value theorem gives  
\(\bigl|\tfrac{f(x\mid a+h)-f(x\mid a)}{h}\bigr|=\lvert f_a(x\mid a+\theta h)\rvert\le g(x)\).  
Because \(g\in L^{1}\), it dominates the entire difference–quotient sequence, so the dominated-convergence theorem implies  
\(\displaystyle\lim_{h\uparrow0}\int s(x)\tfrac{f(x\mid a+h)-f(x\mid a)}{h}\,dx
=\int s(x)\,f_a(x\mid a)\,dx\)  
for any contract \(s\) with \(\mathbb{E}[\,|s(X(a))|\,]<\infty\).%
}
\end{comment}

Before we proceed, we introduce the \textit{first-best} benchmark, in which the principal has full information and directly chooses the first-best effort \(a_{FB}\) that maximizes her expected payoff. 
\begin{definition}[First Best]\label{FB}
The first-best effort, denoted by \(a_{FB}\), is the effort that maximizes the social surplus $\mathbb{E}\bigl[X(a)\mid a\bigr] \;-\; c(a).$
In \textbf{OP}, if there exists a contract \(s(x)\) such that the agent optimally chooses \(a_{FB}\) and obtains expected utility \(E^A(a_{FB},s)=u_0\), then we say that contract \(s(x)\) achieves the first best.
\end{definition}
By Assumption \ref{ass:interchange}-(ii), we always have $a_{FB} < +\infty$. Moreover, by equations \eqref{principalpayoff} and \eqref{agentpayoff}, and the IR constraint  \(E^A(a, s) \geq u_0\), the principal's expected payoff
\[
\begin{aligned}
    E^P(a, s) &=  \mathbb{E}[\,X(a) \mid a\,] - c(a) - E^A(a, s)
    \\&  \leq \mathbb{E}[\,X(a_{FB}) \mid a_{FB}\,] - c(a_{FB}) - u_0.
\end{aligned}
\]
Thus, \(E^P(a, s)\) is bounded above by the first-best payoff, \(\mathbb{E}[\,X(a_{FB}) \mid a_{FB}\,] - c(a_{FB}) - u_0\), with equality if and only if \(s\) achieves the first best. Any contract that achieves the first best is therefore optimal. %If no contract can achieve the first best, we say that there is a \emph{loss of efficiency}.

% If several optimal actions are equally preferable to the principal, the agent then selects one arbitrarily. 
%The principal evaluates all contracts satisfying (IC), (IR), and (LL) and compares the best achievable outcome to the alternative of not hiring the agent.% Thus, the key step—and the most challenging one—is to determine the optimal contract for the principal’s problem, which is the focus of the following sections. 

\section{Limitations of the First-Order Approach as a Relaxation}\label{IFOAV}

Because the IC constraint entails infinitely many inequalities, it is typical to employ the first-order approach (FOA), replacing IC with its first-order condition (RIC, below) to obtain the relaxed problem (\textbf{RP}):
\[
\begin{alignedat}{3}
\textbf{(RP)}\quad & \max_{ a_s,\; s} && E^P( a_s,s),\\[2pt]
\text{subject to}\quad 
& \text{(RIC)} \ \ \; &&   \frac{\partial E^A(a_s, s)}{\partial a} = 0,\\
& \text{(IR)} \ \ \; && E^A( a_s,s) \ge u_0,\\
& \text{(LL)}\ \ \; && s(x) \ge 0, \ \ \forall x\geq 0.
\end{alignedat}
\]

The FOA is intended to construct a relaxation of the original problem by enlarging the constraint set over which the objective is maximized \citep{rogerson1985first}, akin to other relaxation methods in optimization \citep{fisher1981lagrangian, rockafellar2015convex}. A defining feature of a relaxation is that it preserves all feasible solutions  \citep{gorry1972relaxation}. In fixed‐support settings, this relaxation property always holds and one can focus solely on the relaxed problem—without loss of any optimal contract—when the agent’s expected utility under the optimal contract is constrained to be concave or strictly unimodal in effort by exogenous conditions \citep{rogerson1985first, conlon2009two, jiang2024revisiting}.

In shifting‐support settings, however, the FOA no longer provides a genuine relaxation of the original problem. As we demonstrate in the counterexample below, a quota-bonus contract that is optimal for the \textbf{OP} may not even be feasible for the \textbf{RP}, despite the agent’s expected utility under the optimal contract being strictly unimodal in effort.

\begin{example}\label{exp1}
Let the outcome be $X = a + \xi,\ \xi\sim\mathrm{Exp}(1)$, so that 
\[f(x\mid a) \;=\;
\begin{cases}
e^{-(x - a)}, & x \ge a,\\
0, & x < a.
\end{cases}\] The agent chooses effort \(a \geq 0\) and incurs cost 
$c(a)=\frac{a^2}{2}\,.$ The agent's reservation utility is \(u_0=1\).
\end{example}

This example fits exactly within the framework of \citet{oyer2000theory}. In such a simple setting, one might expect that applying the FOA—i.e., solving the \textbf{RP}—would yield  the familiar result that a quota-bonus contract is optimal, as in \citet{park1995incentive}, \citet{kim1997limited}, \citet{oyer2000theory}, and \citet{jiang2024revisiting}. However, applying the FOA in this case instead leads to the opposite conclusion: quota-bonus contracts are, in fact, not even feasible.

Specifically, consider any quota-bonus contract $s(x) = b\, \mathbf{1}_{\{ x \geq q \}}$ with quota $q > 0$ and bonus $b > 0$, where $\mathbf{1}_{\{\cdot\}}$ denotes the indicator function. Suppose that under \textbf{RP}, the agent selects effort $a$ given $s$, and that $a$ and $s$ satisfy the RIC and IR constraints. The agent's expected utility is then given by
\[
E^A(a, s) = \int_0^\infty b\, \mathbf{1}_{\{x \geq q\}}\, e^{-(x - a)}\, \mathbf{1}_{\{x \geq a\}}\, dx - \frac{1}{2} a^2 =
\begin{cases}
b - \tfrac{1}{2} a^2, & a \geq q, \\[2mm]
b\, e^{-(q-a)} - \tfrac{1}{2} a^2, & a < q.
\end{cases}
\]
Since $a$ and $s$ must satisfy RIC, we have $a < q$, because otherwise the right-hand derivative is $\frac{\partial^+}{\partial a} E^A(a, s) = \frac{\partial^+}{\partial a} (b - \tfrac{1}{2} a^2) = -a \leq -q < 0$. For $a < q$, the RIC constraint is equivalent to
\[
\frac{\partial E^A(a, s)}{\partial a} = b\, e^{-(q-a)} - a = 0,
\]
and the IR constraint is equivalent to
\[
b\, e^{-(q-a)} - \tfrac{1}{2} a^2 \;\ge\; u_0 = 1.
\]
Combining these yields a contradiction, i.e.,
\[
a - \tfrac{1}{2} a^2 \;\ge\; 1 \;\iff\;
-\frac{1}{2}(a - 1)^2  \;\ge\;  \frac{1}{2},
\]
which is impossible for any $a \geq 0$.
This implies that \textbf{RP} admits no feasible quota-bonus contracts. Assuming the FOA to be a valid relaxation, one would conclude that quota-bonus contracts are also infeasible in \textbf{OP} and therefore suboptimal, being dominated by linear contracts such as $s_L(x) = (\sqrt{3}-1)x$ (note that \(E^A(a,s_L)=(\sqrt{3}-1)(a+1)-a^2/2\) attains the unique maximum \(1\) at the stationary point \(a_L=\sqrt{3}-1\) and $(a_L,s_L)$ satisfies RIC, IC, IR, and LL constraints; thus \(s_L\) is feasible in both \textbf{RP} and \textbf{OP} and yields \(E^P(a_L,s_L)=2\sqrt{3}-3\)).
 %Under this contract, the agent optimally selects $a_{L} = \sqrt{3}-1$ and attains $E^A(a_{L}, s_L) =u_0$. Correspondingly, the principal’s payoff is $E^P(a_{L}, s_L) = 2\sqrt{3}-3$.—i.e, preserved all feasible contracts of \textbf{(OP)}—

However, we will show that the conclusion reached under the FOA is incorrect. Consider the quota-bonus contract
$\hat{s}(x) = \frac{3}{2}\;\mathbf{1}\{x \ge 1\}$, with quota $1$ and bonus $\frac{3}{2}$.
Under this contract, the agent’s expected utility is 
\begin{equation} \label{eq:agent-utility}
E^A(a,\hat{s}) = \;\int_0^\infty \frac{3}{2} \cdot \mathbf{1}_{\{x \ge 1\}} \cdot  \,e^{-(x - a)} \cdot \mathbf{1}_{\{x \ge a\}} \,dx - \frac{1}{2}a^2 = 
\begin{cases}
\tfrac{3}{2} - \tfrac{1}{2}a^2, & a \ge 1,\\[2mm]
\tfrac{3}{2}e^{-(1-a)} - \tfrac{1}{2}a^2, & a < 1.
\end{cases}
\end{equation}
It is straightforward to verify that $E^A(a,\hat{s})$ is strictly unimodal and attains its unique maximum value $1$ at $\hat{a}=1$. Thus $(\hat{a},\hat{s})$ satisfies IC but violates the RIC, since the right-hand derivative $\frac{\partial^+}{\partial a} E^A(\hat{a}, \hat{s}) = -\hat{a}= -1 \neq 0$, implying that $\hat a$ is not a stationary point. Moreover, the social surplus $\mathbb{E}[X(a)\mid a]-c(a)=a+1-\tfrac{1}{2}a^{2}$ is maximized at $a=1$, so $\hat{a}$ is the first best effort. Because $\hat{s}$ implements the first-best effort while binding participation $E^A(\hat{a},\hat{s})=1 = u_{0}$, it solves \textbf{OP} optimally with principal payoff $E^P(\hat{a},\hat{s})=\tfrac{1}{2}$ , even though it is infeasible under \textbf{RP}.

Moreover, it can be readily verified that the optimal (limited-liability) linear contract is \(s_L = (\sqrt{3}-1)x\).\footnote{In Example~\ref{exp1}, the optimal (limited-liability) linear contract is \(s_L = (\sqrt{3}-1)x\). Consider \(s'(x)=\alpha+\beta x\) with \(\alpha,\beta \geq 0\): for any \((\alpha,\beta)\), \(E^A(a,s')=\alpha+\beta(a+1)-a^2/2\) is uniquely maximized at \(a'=\beta\), giving \(E^A(a',s')=\alpha+\beta+\beta^2/2\), so IR implies \(\alpha\ge 1-\beta-\beta^2/2\). The principal’s payoff at \(a'=\beta\) is \(E^P(a',s')=1-\beta^2-\alpha\); choosing the smallest feasible \(\alpha\) yields \(\alpha^*=\max\{0,\,1-\beta-\beta^2/2\}\) and hence \(E^P(\beta)=1-\beta^2-\alpha^*=\min\{\,1-\beta^2,\,\beta-\beta^2/2\,\}\), which is maximized at \(\beta^*=\sqrt{3}-1\) (where the two branches meet) with \(\alpha^*=0\).} However, $s_L$ fails to achieve the first best: it induces a lower effort, \(a_L = \sqrt{3}-1 \approx 0.732 < 1 = \hat{a}\), and yields a lower principal payoff $E^P(a_{L}, s_L) = 2\sqrt{3}-3 \approx 0.464 < 0.5= E^P(\hat{a}, \hat{s})$. Therefore, contrary to the FOA-based conclusion, direct analysis of \textbf{OP} shows that the quota-bonus contract $\hat{s}$ not only attains optimality but also strictly dominates all linear contracts. This divergence indicates that reliance on the FOA can be misleading—particularly for the type of commission-versus-bonus comparisons, a common focus in the sales-force compensation literature \citep[e.g.,][]{kishore2013bonuses,chung2014bonuses,schottner2017optimal}.

%Such an exclusion of an optimal contract is impossible in fixed-support settings, where the FOA's relaxation is valid and ensures that the principal’s payoff from any contract feasible in \textbf{(OP)} is no greater than the payoff the same contract yields in \textbf{(RP)}. In our example, however, 

This exclusion of the optimal (quota-bonus) contract and the consequent bias occurs because—as shown in Figure~\ref{fig:agent-utility}—the optimal contract  \(\hat{s}(x)\) induces the agent to uniquely select a kink point, an interior, non-differentiable optimum that the FOA fails to capture, despite the agent’s expected utility $E^A(a,\hat{s})$ being strictly unimodal in effort $a$. Therefore, the FOA fails to constitute a genuine relaxation, and its validity for quota-bonus contracts (see Definition \ref{defFOAVQB}) under strict unimodality—a powerful property for solving the principal-agent problem in fixed‐support settings \citep{jiang2024revisiting}—fails to hold in the shifting‐support case. 

%\footnote{The non-differentiability issue never occurs in the fixed-support setting, since in that case, for any $s$, $E^A(a,s) = \int f(x\mid a)s(x)dx - c(a)$ is differentiable for all $a>0$ under the same regularity assumptions. }

These limitations significantly reduce the FOA’s practical relevance and suggest that similar failures may arise for other bonus forms, such as two-tier bonuses \citep{moldovanu2001optimal, wang2021moral}. Motivated by these limitations, we introduce a new, genuine relaxation approach in the next section.

\section{The Implementation Relaxation Approach}\label{IRA}
Combining equations \eqref{principalpayoff} and \eqref{agentpayoff}, the principal’s expected utility can be rewritten as
\[E^P(a,s) =\mathbb{E}[\,X(a) \mid a \,] - c(a) -  E^A(a,s).\]
This reformulation shows that once the agent's effort-utility pair $(a, E^A)$ is given, the principal’s expected payoff is determined. Before proceeding, we formally define the notion of ``implement" and the ``implementable set" below. For notational convenience, we will use $u \equiv E^A$ in what follows.
\begin{definition}[Implementation]\label{def:implementation}
A contract \( s \) is said to \emph{implement} an effort-utility pair \( (a, u) \) if \( (a, s) \) satisfies the IC, IR, and LL constraints and under \( s \), the agent optimally chooses action \( a \geq 0 \) and receives expected utility \( u(a, s) \equiv E^A(a, s) \geq u_0 \). A pair \((a,u)\) is \emph{implementable} if there exists some contract \(s\) that implements it. \emph{The implementable set}, denoted by $\mathcal{I}$, is the collection of all implementable pairs:
\[
  \mathcal{I}
  \;\triangleq\;
  \bigl\{(a,u)\mid (a,u)\text{ is implementable}\bigr\}.
\]
\end{definition}

Given the reformulation above, if the principal can ex ante identify the \textit{optimal implementable pair}, i.e., the effort-utility pair that maximizes her expected payoff across all implementable pairs, then determining the optimal contract reduces to designing any contract that implements this pair.
%In particular, to establish the optimality of quota-bonus contracts, it suffices to verify that a quota-bonus contract can implement the optimal implementable pair.  % rather than solving the full infinite-dimensional bi-level optimization problem. %As Example~\ref{exp1} demonstrates, the quota-bonus contract is optimal because it implements the optimal implementable pair \((a_{FB},u_0)\), where $a_{FB}$ is the first best effort. 

%Since it maximizes the principal’s expected payoff \(E^P(a,s)\) over all \((a,u)\in[\underline a,\bar a]\times[u_0,\infty)\).

%By definition, the first best is achievable if and only if \((a_{FB},u_0)\) is implementable.  
%However, in general, when there is a loss of efficiency (i.e., the first best cannot be achieved), no contract can implement \((a_{FB},u_0)\) 
The primary challenge, then, is how to pinpoint the optimal implementable pair.
Given the implementable set $\mathcal{I}$, the task of finding the optimal implementable pair reduces to the following program:
\begin{equation}\label{true}
    \max_{(a,u) \in \mathcal{I}} \;   \mathbb{E}[\,X(a) \mid a \,] - c(a) - u.
\end{equation}
This formulation reduces the search for the optimal implementable pair to a simple two-dimensional optimization. If $\mathcal{I}$ can be explicitly characterized, identifying the optimal implementable pair becomes straightforward. However, characterizing the implementable set \(\mathcal{I}\) is generally challenging. A naive brute‐force approach, which tests for each \((a,u)\in[0,\infty]\times[u_0,\infty]\) whether some contract implements it, is clearly impractical; it would be no simpler than directly solving the original infinite‐dimensional bi‐level program.
%Therefore, the primary challenge in finding the optimal $(a, u)\in \mathcal{I}$ lies in determining the structure of the implementable set $\mathcal{I}$.

\begin{comment}
    However, characterizing the implementable set \(\mathcal{I}\) can be difficult in general. To address this, we propose an \emph{implementation relaxation approach} (IRA). The core idea is to replace \(\mathcal{I}\) with a tractable superset \(\mathcal{A}\) satisfying \(\mathcal{I}\subseteq\mathcal{A}\). We then solve the relaxed program
\[
\max_{(a,u)\in\mathcal{A}}
\int_{L(a)}^{\infty} x\,f(x\mid a)\,dx
\;-\;c(a)\;-\;u,
\]
obtaining a maximizer \((a_{\mathcal{A}},u_{\mathcal{A}})\in\mathcal{A}\). If this candidate also lies in \(\mathcal{I}\), it must maximize the firm’s expected payoff over all \((a,u)\in\mathcal{I}\), since \(\mathcal{I}\subseteq\mathcal{A}\).
\end{comment}

Instead of exhaustively characterizing \(\mathcal{I}\), we introduce the \emph{implementation relaxation approach} (IRA). The central idea of the IRA is that when $\mathcal{I}$ is difficult to determine, one can instead consider a relaxed, simpler superset $\mathcal{I}^{\mathrm{rel}}$ such that 
\[\mathcal{I} \subseteq \mathcal{I}^{\mathrm{rel}}.\]
For a given $\mathcal{I}^{\mathrm{rel}}$, we solve the following relaxed program:
\begin{equation}\label{RP}
\max_{(a,u) \in \mathcal{I}^{\mathrm{rel}}} \; \mathbb{E}[\,X(a) \mid a \,] - c(a) - u, 
\end{equation}
yielding a \textit{relaxed optimal effort-utility pair} $(a^{\mathrm{rel}}, u^{\mathrm{rel}}) \in \mathcal{I}^{\mathrm{rel}}$, i.e., a solution to \eqref{RP}.

Since $\mathcal{I} \subseteq \mathcal{I}^{\mathrm{rel}}$, the pair $(a^{\mathrm{rel}}, u^{\mathrm{rel}})$ yields a principal’s expected payoff at least as high as any effort-utility pair in $\mathcal{I}$. Thus, if $(a^{\mathrm{rel}}, u^{\mathrm{rel}})$ is also implementable, it is precisely the optimal implementable pair we seek.

Therefore, to identify the optimal implementable pair, it suffices to find an optimal pair in a \textit{relaxed implementable set} (i.e., a superset of the implementable set), and then verify that the resulting pair is actually implementable. Once a relaxed optimal effort-utility pair $(a^{\mathrm{rel}}, u^{\mathrm{rel}}) \in \mathcal{I}^{\mathrm{rel}}$ is obtained, it is enough to show that there exists a contract that implements $(a^{\mathrm{rel}}, u^{\mathrm{rel}})$; such a contract is then optimal. In other words, a sufficient condition for a contract to be optimal is that it implements a relaxed optimal effort-utility pair $(a^{\mathrm{rel}}, u^{\mathrm{rel}})$ from a relaxed implementable set $\mathcal{I}^{\mathrm{rel}}$.

The most trivial relaxed implementable set is simply $[0,\,\infty] \times [u_0,\,\infty]$. If we set $\mathcal{I}^{\mathrm{rel}} = [0,\,\infty] \times [u_0,\,\infty]$, the relaxed optimal effort-utility pair is always $(a_{FB},\,u_0)$. This relaxation is innocuous when the first best can be achieved. However, in cases where there is a \textit{loss of efficiency} i.e., first best cannot be achieved, it is overly ``lax" and fails to provide useful information about the optimal implementable pair. Intuitively, the ``tighter" the relaxation, i.e., the smaller the gap between \(\mathcal{I}^{\mathrm{rel}}\) and \(\mathcal{I}\), the more reliably the relaxed program \eqref{RP} will recover the true optimal implementable pair. The following theorem provides a ``tighter" relaxation. 
%\footnote{For notation convenience, we include the boundary point $x = L(a)$, where $f$ may be non-differentiable, by defining $f_a(L(a)\mid a) \equiv 0$. Throughout, we also allow the supremum, i.e. $\sup (\cdot)$, to be $+\infty$. }

\begin{theorem}\label{nececondimp}
Define
\[
\mathcal{J} \triangleq 
\left\{ (a, u) \;\middle|\; a > 0,\, u \geq u_0,\, 
\sup_{x \in (L(a),\, \bar{x}]} \frac{f_a(x \mid a)}{f(x \mid a)} \geq \frac{c'(a)}{c(a) + u}
\right\}
\,\cup\,
\left\{ (0, u) \mid u \geq u_0 \right\}.
\]
Then:
\begin{enumerate}
    \item Any implementable effort-utility pair $(a, u)$ belongs to $\mathcal{J}$, so that $\mathcal{I} \subseteq \mathcal{J}$.
    \item There exists $(a, u) \in \mathcal{J}$ that solves
    \begin{equation}\label{RPJ}
        \max_{(a, u) \in \mathcal{I}} \left\{ \mathbb{E}[X(a) \mid a] - c(a) - u \right\}.
    \end{equation}
    \item Let \((a^\mathcal{J}, u^\mathcal{J})\) be a maximizer of \eqref{RPJ}. Then any contract that implements \((a^\mathcal{J}, u^\mathcal{J})\) is optimal for \textbf{OP}.
\end{enumerate}
\end{theorem}

\begin{comment}
    \begin{theorem}\label{nececondimp}
Any implementable effort-utility pair \((a, u)\) must satisfy \((a, u) \in \mathcal{J}\), where  
\[
\mathcal{J}
\triangleq 
\left\{(a,u)\;\middle|\;  a > 0,\; u \ge u_0,\;
\sup_{x \in (L(a), \bar x]} \left( \frac{f_a(x\mid a)}{f(x\mid a)} \right)
\ge \frac{c'(a)}{c(a)+u}
\right\}
\;\cup\;
\Big\{(0,u)\mid u\ge u_0 \Big\}.
\]
Moreover, there always exists $(a, u)\in \mathcal{J}$ that solves 
\begin{equation}\label{RPJ}
\max_{(a,u) \in \mathcal{J}}  \mathbb{E}[\,X(a) \mid a \,] - c(a) - u.
\end{equation}
Let $(a^\mathcal{J}, u^\mathcal{J})$ denote a maximizer of \eqref{RPJ}. Then any contract that implements $(a^\mathcal{J}, u^\mathcal{J})$ is optimal for the principal’s problem. 
\end{theorem}
\end{comment}

The relaxed implementable set \(\mathcal{J}\) naturally consists of two parts. First, any pair \((0, u)\) with \(u \ge u_0\) is implementable, since \(a=0\) is the minimal effort level and the principal can simply offer a fixed payment to ensure participation. Second, for \(a > 0\), the pair \((a,u)\) is implementable only if
\(\sup_{x\in (L(a), \bar x]}\big(\frac{f_a(x\mid a)}{f(x\mid a)}\big)
\ge\frac{c'(a)}{c(a)+u}\) 
i.e.\ only if the supremum likelihood ratio (may be $+\infty$) is at least the value of the agent’s marginal disutility per unit of total compensation. Intuitively, $f_a(x\mid a)$
measures how a marginal increase in effort raises the likelihood of observing outcome \(x\), and thus the likelihood ratio $\frac{f_a(x\mid a)}{f(x\mid a)} = \frac{s(x)f_a(x\mid a)}{s(x)f(x\mid a)}$ represents the marginal reward per unit of compensation at $x$. If this inequality fails, then for no outcome \(x\) does the payment make the marginal reward sufficient to offset the marginal disutility of effort; consequently, regardless of how compensation is structured across outcomes, the incentive compatibility cannot hold.

%By contrast, when the inequality holds, a contract \(s\) that implements \((a,u)\) may exist (though not guaranteed): the principal can concentrate payments at the output level where the likelihood ratio peaks, thereby balancing the marginal reward for effort with its marginal disutility.
%Note, however, that this condition is necessary but not sufficient for implementability.

%\\ A similar argument to this lemma was established by \citet{jiang2024revisiting} (Lemma 2) in fixed-support settings under the monotonically increasing likelihood ratio assumption with risk-neutral agents. In that case, \( E^A(a,s) \) is always differentiable in \( a \), and for contracts that implement \( (a,u) \), it satisfies \( E^A_a(a,s) \equiv 0 \). Our Lemma~\ref{nececondimp} generalizes this result by allowing any likelihood ratio behavior, allowing the agent to be either risk-neutral or risk-averse, and extending it to shifting-support settings. Additionally, in \citet{jiang2024revisiting}, the inequality is strict because the likelihood ratio is assumed to be strictly increasing, and the density function is continuous in \( x \) with no point masses. 

%We will demonstrate in Section \ref{clrp} that under specific conditions on the likelihood ratio of the outcome distribution, the implementable set $\mathcal{I}$ is identical to  $\mathcal{J}$.

 Therefore, our relaxed implementable set $\mathcal{J}$ captures the necessary condition implied by both IC and IR for implementable effort-utility pairs. This relaxation, by construction, preserves all feasible contracts of \textbf{OP} as each feasible contract is recoverable by a pair in $\mathcal{J}$. To conclude this section, we present an example where the optimal (quota-bonus) contract induces kink optima that the IRA captures, unlike the FOA. The accompanying loss of efficiency (identified via $\mathcal{J}$) underscores the relaxation’s tightness; notably, in this case the relaxation is ``perfectly tight", as the true implementable set \(\mathcal{I}\) coincides with \(\mathcal{J}\). %and, we can find an optimal contract that incentivizes the agent to uniquely select a kink point point—a solution not feasible under the relaxed problem induced by the FOA.in contrast to the FOA, which imposes a necessary condition solely on the agent’s effort at differentiable points under (IC)  and leave the discussion of the optimality of quota-bonus contracts to the next section

\begin{example}[loss of efficiency]\label{pareto}
The outcome $X$ follows a Pareto distribution \[f(x\mid a) \;=\;
\begin{cases}
\frac{2a^2}{x^3}, & x \ge a,\\
0, & x < a.
\end{cases}\] The agent chooses effort $a \geq 0$ and incurs cost $c(a)=a^3$. The agent's reservation utility is $u_0=0$.
\end{example}

In this example, \(\sup_{x \in (L(a), \bar x]} \big(\frac{f_a(x\mid a)}{f(x\mid a)}\big)
\ge\frac{c'(a)}{c(a)+u}\) is equivalent to $\frac{2}{a} \geq \frac{3a^2}{a^3+u}$, which is equivalent to $u\geq \frac{a^3}{2}$. Hence, the relaxed implementable set $\mathcal{J} = \{ (a,u) \mid a \geq 0,\; u \ge \tfrac{a^3}{2} \}$, and program \eqref{RPJ} is equivalent to $\max_{u\geq\, \frac{a^3}{2} \, \geq 0} \; \{2a - a^3 - u\}$. Solving this yields a unique solution $(a^\mathcal{J}, u^\mathcal{J}) = (\frac{2}{3}, \frac{4}{27})$. One can easily verify that the quota-bonus contract $s^*(x)=\tfrac{4}{9} \cdot \mathbf{1}_{\{x\ge2/3\}}$
implements \((a^\mathcal{J},u^\mathcal{J})\). Indeed, under \(s^*\) the agent’s expected utility is
$\mathbb{E}[s^*(X)\mid a]-c(a)
=\min\{a^2,\tfrac{4}{9}\}-a^3$,
which attains its maximum uniquely at a kink point \(a^\mathcal{J}=\tfrac{2}{3}\).
Therefore, by Theorem \ref{nececondimp}, $s^*(x)$ is the optimal contract. 

It is worth noting that the first best cannot be achieved because $(a_{FB}, u_0) = (\frac{\sqrt{6}}{3},0) \notin \mathcal{J}$. Instead, the agent's optimal effort is $a^\mathcal{J} = \frac{2}{3} < a_{FB}$, and his corresponding expected utility is $\frac{4}{27} > u_0$. The principal's expected payoff under $s^*(x)$ is $\frac{8}{9}$, which is lower than the principal's first-best expected payoff $\frac{4\sqrt{6}}{9}$.  

Moreover, for any $(\hat a, \hat u)\in \mathcal{J}$, one can choose $s = (\hat a^3 + \hat u) \cdot\mathbf{1}_{\{x\ge \hat a\}}$ to implement it. Under $s$, the agent's expected utility is $\mathbb{E}[s(X(a))\mid a] - c(a) =(\hat a^3 + \hat u)\, \min \{ \frac{a^2}{\hat a^2}, 1\} \;-\; a^3$. For $a\geq \hat a$, $(\hat a^3 + \hat u)\, \min \{ \frac{a^2}{\hat a^2}, 1\} \;-\; a^3$ is strictly decreasing in $a$. For $0\leq a< \hat a$, its derivative is $(\hat a^3 + \hat u)\frac{2a}{\hat a^2} - 3a^2$, which is non-negative because $(\hat a, \hat u)\in \mathcal{J}$ implies $\hat u \ge \frac{\hat a^3}{2} \implies (\hat a^3 + \hat u)\frac{2a}{\hat a^2} \geq 3 \hat a a \geq 3a^2$. Hence, the agent optimally chooses \(a=\hat a\), yielding the agent's expected utility \(\hat u\). This demonstrates that any $(\hat a, \hat u)\in \mathcal{J}$ is implementable, combining that $\mathcal I \subseteq \mathcal{J}$ (by Theorem \ref{nececondimp}), we conclude  $\mathcal I = \mathcal{J}$.

 %it should satisfy $\frac{\hat{a}^2}{a^2} \leq \frac{c(\hat{a})+u}{c(a)+r} = \frac{\hat{a}^3+r}{a^3+r}, \forall \hat{a}<a$. Note that $\frac{\hat{a}^2}{a^2} \leq \frac{\hat{a}^3+r}{a^3+r}, \forall \hat{a}<a$ is equivalent to $r \geq \frac{a^2\hat{a}^2(a-\hat{a})}{a^2-\hat{a}^2} = \frac{a^2\hat{a}^2}{a+\hat{a}},\forall \hat{a}<a$. Since $\frac{a^2\hat{a}^2}{a+\hat{a}}$ is increasing in $\hat{a}$ for $\hat{a}<a$, we conclude that $r \geq \frac{a^2\hat{a}^2}{a+\hat{a}},\forall \hat{a}<a$ if and only if $r \ge \tfrac{a^3}{2} $. Hence, the set of implementable pairs is given by  $\mathcal{I} = \{ (a,u) \mid a \geq 0,\; r \ge \tfrac{a^3}{2} \}$. 

%\textcolor{red}{Provide the high-level logic of relaxing the implementable set here. Then introduce the following sections of CLRP (where $\mathcal{I}$ can be identified) and beyond CLRP.}

%\textcolor{red}{put it somewhere else} 

%In contrast, our approach relaxes the implementable set. \textcolor{red}{Discuss how IRA coincides the philosophy of FOA.}

%This avoids the need to directly solve the original infinite-dimensional, bi-level optimization problem and also makes it unnecessary to consider a relaxed problem, as is typical in the literature.

\section{Optimality Conditions for Quota-bonus Contracts}\label{OCQBC}
%In the principal-agent literature with limited liability and risk neutrality, it is of particular interest to determine when quota-bonus contracts are optimal.In particular, \cite{park1995incentive}`and \cite{kim1997limited} focus on the fixed-support setting and invoke the celebrated Mirrlees-Rogerson condition to justify  By Theorem \ref{nececondimp}, such a pair can always be identified.

In the preceding section, the IRA implies that a sufficient condition for contract optimality is the implementation of a relaxed optimal effort-utility pair. This yields a direct application of the IRA: to demonstrate the optimality of a simple contract form, it suffices to verify whether any contract within that parametric form implements a relaxed optimal effort-utility pair. Simple contracts have long been studied in the literature, with quota-bonus contracts receiving particular attention under the canonical model of risk neutrality and limited liability \citep{park1995incentive, kim1997limited, oyer2000theory, jiang2024revisiting}. In this section, we focus on quota-bonus contracts and employ the IRA to derive their optimality conditions. We then compare these conditions with the FOA-based conditions established in prior work to assess the effectiveness of the IRA.

%including the celebrated Mirrlees-Rogerson condition and the refinements in \cite{jiang2024revisiting}.

%condition. That is, whenever there exists a quota-bonus contract that optimally solves the relaxed problem and also solves the principal's original problem, our conditions are necessarily satisfied.

%In this section, we establish optimality conditions for quota-bonus contracts using our Implementation Relaxation Approach.

%, and thus allows them to derive weaker exogenous conditions for the optimality of quota-bonus contract.
%Contrasting these literature establishing the optimality conditions using first-order approach for quota-bonus contracts.

%Since Theorem~\ref{nececondimp} allows us to identify a relaxed optimal implementable pair \((a^\mathcal{J}, u^\mathcal{J})\), a sufficient condition for the optimality of quota-bonus contracts is the existence of a quota-bonus contract that implements \((a^\mathcal{J}, u^\mathcal{J})\). 

Quota-bonus contracts form a class of simple contracts defined by \( s_q = b\, \mathbf{1}_{\{x \ge q\}} \), where \( q, b \in [0, +\infty) \). Hence, identifying quota-bonus contracts that implement a specific effort utility pair amounts to searching over appropriate values of $q$ and $b$. The following lemma further narrows the search.

%The following Lemma reduces to considering quota-bonus contracts as specified as below.

%determining the appropriate quota \(q\), since the bonus \(b\) is then uniquely determined by the (IR) constraint. The selection of \(q\) is further determined by the (IC) constraint.

\begin{lemma}\label{optimalqb}
Suppose the quota-bonus contract \(s_q = b\, \mathbf{1}_{\{x \ge q\}}\) implements the pair \((a^\mathcal{J},u^\mathcal{J})\). Then the bonus and quota must satisfy the following: 
the bonus is given by \(b = \frac{c(a^\mathcal{J}) + u^\mathcal{J}}{1 - F(q \mid a^\mathcal{J})}\), and the quota \(q\) satisfies either \(q = L(a^\mathcal{J})\), or  \(q > L(a^\mathcal{J})\) and $\frac{-F_a(q \mid a^\mathcal{J})}{1 - F(q \mid a^\mathcal{J})} = \frac{c'(a^\mathcal{J})}{c(a^\mathcal{J}) + u^\mathcal{J}}$.
\end{lemma}

%Lemma~\ref{optimalqb} further narrows the search, reducing it to very few, or even just one or two, candidate contracts. For example, when the likelihood ratio is strictly increasing, \(\frac{-F_a(q \mid a^{\mathcal J})}{1 - F(q \mid a^{\mathcal J})}\) is strictly increasing in \(q\) \citep[Lemma 1(ii)]{jiang2024revisiting}. This guarantees that the equation $\frac{-F_a(q \mid a^{\mathcal J})}{1 - F(q \mid a^{\mathcal J})} = \frac{c'(a^{\mathcal J})}{c(a^{\mathcal J}) + u^{\mathcal J}}$ has at most one solution and thus it suffices to consider at most two candidate contracts,

Notably, the bonus $b$ is uniquely determined by the IR constraint once the quota \(q\) is given.
It is sufficient to consider only these few candidate contracts because the IC constraint requires that \(a^{\mathcal J}\) be either a kink point, a boundary point, or a stationary point of the agent’s expected utility under any contract that implements \((a^{\mathcal J}, u^{\mathcal J})\). Note that the agent’s expected utility under \(s_q = b\, \mathbf{1}_{\{x \ge q\}}\) is \(b(1-F(q\mid a))- c(a)\), which is differentiable in \(a\) when \(q > L(a^{\mathcal J})\). Thus, \(q > L(a^{\mathcal J})\) satisfying $\frac{-F_a(q \mid a^{\mathcal J})}{1 - F(q \mid a^{\mathcal J})} = \frac{c'(a^{\mathcal J})}{c(a^{\mathcal J}) + u^{\mathcal J}}$ corresponds to the case where \(a^{\mathcal J}\) is a stationary point. By contrast, \(q = L(a^{\mathcal J})\) corresponds to the case where \(a^{\mathcal J}\) is a kink point or a boundary point (i.e., \(a^{\mathcal J} = 0\)). We thus establish the following exogenous conditions.

\begin{theorem}[Optimality of quota-bonus contracts]\label{qboptimalcondition}
Let \(\bigl(a^{\mathcal J},u^{\mathcal J}\bigr)\) be any solution to program~\eqref{RPJ}.  
There exists a quota-bonus contract that optimally solves \textbf{OP} if one of the following holds:

\noindent\textbf{(C1)}\;
\(
1-F (L(a^{\mathcal J})\mid a )
\le
\frac{c(a)+u^{\mathcal J}}{c(a^{\mathcal J})+u^{\mathcal J}},
\ \ \forall\,a \in [0, a^{\mathcal J}),
\)
\begin{center}
or
\end{center}
\noindent\textbf{(C2)}\;
there exists \(q>L(a^{\mathcal J})\) satisfying 
\(
\frac{-F_a(q\mid a^{\mathcal J})}{1-F(q\mid a^{\mathcal J})}
=
\frac{c'(a^{\mathcal J})}{c(a^{\mathcal J})+u^{\mathcal J}}
\)
such that 
\(
\frac{1-F(q\mid a)}{1-F\!(q\mid a^{\mathcal J})}
\le
\frac{c(a)+u^{\mathcal J}}{c(a^{\mathcal J})+u^{\mathcal J}},
\ \ \forall a\geq 0.
\)
\end{theorem}

Checking either \textbf{(C1)} or \textbf{(C2)} is straightforward: once the distribution \(F(\cdot\mid a)\) and the cost function \(c(\cdot)\) are specified, for \textbf{(C1)}, each side of the inequality is a single‐valued function of \(a\) and verifying optimality therefore reduces to a pointwise comparison of two scalar functions. For \textbf{(C2)}, one first solves 
\(
\frac{-F_a(q\mid a^{\mathcal J})}{1-F(q\mid a^{\mathcal J})}
=\frac{c'(a^{\mathcal J})}{c(a^{\mathcal J})+u^{\mathcal J}}
\)
for \(q>L(a^{\mathcal J})\), a straightforward one‐dimensional root‐finding step, and then simply compares two scalar functions of \(a\).  It is also worth noting that our argument accommodates cases where program~\eqref{RPJ} admits multiple solutions. In such instances, $(a^{\mathcal{J}},u^{\mathcal{J}})$ may be taken as any solution to ~\eqref{RPJ}.
Moreover, Theorem \ref{qboptimalcondition} shows that our exogenous conditions not only capture kink optima but also guarantee the existence of an optimal quota-bonus contract—features absent from FOA-based conditions \citep{jiang2024revisiting}. %Assuming the existence of an optimal contract, we will show in the following that our exogenous conditions are even more general than the endogenous FOA-validity condition for quota-bonus contracts, and therefore more general than any FOA-based exogenous conditions. 

\subsection{Comparison with FOA-Based Conditions}
Literature on the optimality of quota-bonus contracts typically relies on a common assumption that the FOA is valid for quota-bonus contracts \citep{park1995incentive, kim1997limited, oyer2000theory, jiang2024revisiting}. Following \cite{jiang2024revisiting}, we give the following definition. 
\begin{definition}\label{defFOAVQB}
The FOA is valid for quota-bonus contracts (FOAVQB) if, whenever a quota-bonus contract solves \textbf{RP}, it also solves \textbf{OP}; and if no quota-bonus contract solves \textbf{RP}, then none solves \textbf{OP}.
\end{definition}

The FOAVQB is endogenous and, in fixed-support settings, is justified by the exogenous conditions in \cite{jiang2024revisiting} or by the Mirrlees-Rogerson condition \citep{rogerson1985first, park1995incentive}. The following theorem enables a comparison of our conditions with FOAVQB and with any FOA-based exogenous conditions, in terms of the optimality of quota-bonus contracts.

%that although neither these conditions nor the endogenous FOAVQB ensure the existence of an optimal contract, our conditions are more general whenever a optimal quota-bonus contract exists.

%The following theorem demonstrates that the FOA-validity for quota-bonus contracts implies condition \textbf{(C2)}.

\begin{theorem}\label{weakerthanFOAVQB}
Suppose there exists a quota-bonus contract $s$ that optimally solves \textbf{RP} and also solves \textbf{OP}. Then condition \textbf{(C2)} must hold.
\end{theorem}

%Theorem~\ref{weakerthanFOAVQB} implies that every scenario in which a quota-bonus contract solves both \textbf{RP} and \textbf{OP} is covered by \textbf{(C2)} (and hence by \textbf{(C1)} $\lor$ \textbf{(C2)}). To compare with the endogenous FOAVQB, partition the scenarios in which a quota-bonus contract is optimal for \textbf{OP} into two disjoint classes: (i) the contract is also optimal for \textbf{RP}; and (ii) it is optimal only for \textbf{OP}. FOAVQB pertains only to class (i). By Theorem~\ref{weakerthanFOAVQB}, all cases in class (i) fall under \textbf{(C2)}, whereas Example~\ref{exp1} belongs to class (ii) and is covered by \textbf{(C1)}. Therefore, \textbf{(C1)} $\lor$ \textbf{(C2)} is strictly more general than FOAVQB with respect to the optimality conditions for quota-bonus contracts and encompasses broader practical scenarios. Consequently, FOA-based exogenous optimality conditions in the literature \citep{park1995incentive,kim1997limited,jiang2024revisiting} necessarily address only class (i): they first ensure that a quota-bonus contract optimally solves \textbf{RP} and then impose additional constraints so that it also solves \textbf{OP}.Thus, efforts to justify the FOA in shifting-support settings as a supplement to \citet{oyer2000theory}’s FOA-based optimality result for quota-bonus contracts are entirely redundant: the IRA method renders such justifications obsolete by providing a strictly more general foundation.

Theorem~\ref{weakerthanFOAVQB} implies that every scenario in which a quota-bonus contract solves both \textbf{RP} and \textbf{OP} is covered by \textbf{(C2)} (and hence by \textbf{(C1)} $\lor$ \textbf{(C2)}).
To compare with the endogenous FOAVQB, partition the scenarios in which a quota-bonus contract is optimal for \textbf{OP} into two disjoint classes: (i) those in which the contract is also optimal for \textbf{RP}; and (ii) those in which it is optimal only for \textbf{OP}.
FOAVQB pertains only to class (i). By Theorem~\ref{weakerthanFOAVQB}, all cases in class (i) fall under \textbf{(C2)}, whereas Example~\ref{exp1} belongs to class (ii) and is covered by \textbf{(C1)}.
Therefore, \textbf{(C1)} $\lor$ \textbf{(C2)} is strictly more general than FOAVQB with respect to the optimality conditions for quota-bonus contracts and encompasses a broader range of practical scenarios.
Accordingly, FOA-based exogenous optimality conditions in the literature \citep{park1995incentive,kim1997limited,jiang2024revisiting} address only class (i): they first ensure that a quota-bonus contract optimally solves \textbf{RP} and then impose constraints so that it also solves \textbf{OP}.
This suggests that attempts to justify the FOA in shifting-support settings as a supplement to \citet{oyer2000theory}’s FOA-based optimality result for quota-bonus contracts are entirely redundant: the IRA method renders such justifications obsolete by providing a strictly more general foundation.

\section{Conclusions}
This paper revisits the principal-agent problems with moral hazard. While the FOA has become the standard tool for relaxing the incentive-compatibility constraint, we demonstrate that it fails as a relaxation in shifting-support settings. The core problem arises from the presence of non-differentiable optima such as kink solutions, which are excluded by the first-order condition. As a result, optimal (quota-bonus) contracts that are feasible and optimal in the original problem may be lost in the FOA-based relaxation. This finding casts doubt on the reliability of the FOA in contexts where non-differentiability naturally arises, and may be misleading in applications such as commission-versus-bonus comparisons. 

Motivated by this limitation, we introduced the Implementation Relaxation Approach (IRA), which relaxes the set of implementable effort-utility pairs rather than the incentive-compatibility constraint. The IRA not only preserves all feasible contracts but also is simple and general enough to apply, particularly for studying the optimality of simple contracts. Using the IRA, we derived an exogenous optimality condition for quota-bonus contracts that applies in both fixed- and shifting support settings, filling a gap where exogenous optimality
conditions are missing in shifting support settings. Our condition not only ensures the existence of optimal contracts but is also strictly more general than all FOA-based optimality conditions. 

Looking ahead, beyond quota-bonus contracts, the IRA can be readily extended to other simple contractual forms—such as linear, debt, or quadratic—by appropriately parameterizing the problem, following the same procedure employed in our analysis. Moreover, our analysis can also be adapted to discrete settings by replacing the density with point–mass measures.

Furthermore, two directions appear particularly promising. First, beyond the canonical model studied here, IRA could be applied to a broader class of principal-agent problems, such as multitasking agency \citep{dai2021incentive} or dynamic contracting \citep{plambeck2006partnership}. Second, although the exogenous superset we use to approximate the true implementable set appears sufficiently tight for the quota-bonus case, it remains an open question whether tighter relaxations can be developed. We believe that IRA’s relaxation of the implementable set has the potential to become a methodological tool as general as FOA’s relaxation of the IC constraint. Accordingly, we expect IRA to serve as a versatile framework for advancing the analysis of principal-agent problems, and potentially other bilevel optimization problems.

%\THEEndNotes
%\begingroup \parindent 0pt \parskip 0.0ex \def\enotesize{\normalsize} \theendnotes \endgroup
%there are two directions to , one is applying the IRA 

%This strengthens the theoretical foundation for studying quota-bonus contracts, 
%a contract form widely observed in practice, and shows that their optimality does not hinge on restrictive exogenous assumptions.  

%it provides a systematic way to verify the optimality of contracts of given parametric forms, making it particularly suitable for studying simple but practically relevant contracts such as quota-bonus schemes.

% Head 1

%\THEEndNotes
%\begingroup \parindent 0pt \parskip 0.0ex \def\enotesize{\normalsize} \theendnotes \endgroup

% Start of "Sample References" section

% In the interest of anonymization, please do not include acknowledgements in your submission.
%Suppose that, in fixed-support settings, there exists a quota-bonus contract $s$ that optimally solves \textbf{RP} and also solves \textbf{OP}; then condition \textbf{(C2)} must hold.

%\begin{acks}
%
%	The authors would like to thank Dr. Maura Turolla of Telecom
%	Italia for providing specifications about the application scenario.
%
%	The work is supported by the \grantsponsor{GS501100001809}{National
%		Natural Science Foundation of
%		China}{http://dx.doi.org/10.13039/501100001809} under Grant
%	No.:~\grantnum{GS501100001809}{61273304\_a}
%	and~\grantnum[http://www.nnsf.cn/youngscientsts]{GS501100001809}{Young
%		Scientsts' Support Program}.
%
%
%\end{acks}

% Bibliography
\bibliographystyle{informs2014} 
\bibliography{sample-bibliography}

\newpage
\begin{APPENDICES}\label{appn1}
\section{Proofs}

\label{sec:appendix_main_proof}

\setcounter{equation}{0} \numberwithin{equation}{section}

\setcounter{lemma}{0} \numberwithin{lemma}{section}

\setcounter{table}{0} \numberwithin{table}{section}

\begin{lemma}\label{suppositive}
 For any $a>0$,   \[\sup_{x \in (L(a), \bar x]} \left( \frac{f_a(x\mid a)}{f(x\mid a)} \right) > 0.\]
\end{lemma}

\begin{proof}{\textbf{Proof of Lemma \protect\ref{suppositive}}}
~

Suppose not. Then for any $x \in (L(a), \bar x]$, $\frac{f_a(x\mid a)}{f(x\mid a)} \leq 0$, which implies that $f_a(x\mid a)\leq 0, \forall x \in (L(a), \bar x]$. This implies that, for a given $x_0 \in (L(a), \bar x)$,
\[   \int_{L(a)}^{\bar x}  \lim_{h \uparrow 0}  \,\Big[ \frac{f(x\mid a+h)-f(x\mid a)}{h} \Big] \cdot  \mathbf{1}_{x>x_0}\,dx =\int_{x_0}^{\bar x} f_a(x\mid a)dx \leq 0 ,  \] 
which, by Assumption \ref{ass:interchange}-(iv), further implies that 
\[ \lim_{h \uparrow 0}  \int_{L(a)}^{\bar x}  \,\Big[ \frac{f(x\mid a+h)-f(x\mid a)}{h} \Big]  \cdot \mathbf{1}_{x>x_0} \,dx =   \int_{L(a)}^{\bar x} \lim_{h \uparrow 0}  \,\Big[ \frac{f(x\mid a+h)-f(x\mid a)}{h} \Big] \cdot \mathbf{1}_{x>x_0}  \,dx \leq 0.  \] 
It follows that
\[\begin{aligned}
0 \geq     \lim_{h \uparrow 0}  \int_{L(a)}^{\bar x}  \,\Big[ \frac{f(x\mid a+h)-f(x\mid a)}{h} \Big]  \cdot \mathbf{1}_{x>x_0} \,dx
&= \lim_{h \uparrow 0}  \int_{x_0}^{\bar x} \Big[ \frac{f(x\mid a+h)-f(x\mid a)}{h} \Big] dx
\\&= \lim_{h \uparrow 0}  \frac{-F(x_0 \mid a + h) + F(x_0 \mid a)}{h} 
\\& = - F_a(x_0 \mid a).
\end{aligned} 
   \] However, that $- F_a(x_0 \mid a) \leq 0$ contradicts the model assumption that $F_a < 0$ for any $(x, a) \in  (L(a),\bar x] \times (0, +\infty)$. 
$\hfill$ $\blacksquare$
\end{proof}
~

\begin{proof}{\textbf{Proof of Theorem \protect\ref{nececondimp}}}
~

\textbf{(i)} By definition, any implementable $(a, u)$ must satisfy $a \geq 0$ and $u\geq u_0$. Also, for any $ u  \geq u_0$, $(0, u)$ is implemented by contract $s(x) = c(0) + u$. Hence, to show that any implementable $(a, u)$ must belong to $\mathcal{J}$, it remains to show that for any implementable pair $(a, u)$ satisfying  $a>0$ and $u\geq u_0$, we must have 
\begin{equation}\label{nececondition}
    \sup_{ x \in (L(a), \, \bar x]} \frac{f_a(x \mid a)}{f(x \mid a)} \geq \frac{c'(a)}{c(a) + u}.
\end{equation}
Since \((a,u)\) is implementable, there must exist a limited liability contract \(s(x)\) such that \(E^A(a,s) = u\) and \(E^A(\hat{a},s) \leq u\) for all $\hat{a}> 0$. Hence, that \(a > 0 \) is a global maximum implies that for any $h \in (-a,0)$, \( E^A(a+h,s) -E^A(a,s) \leq 0 \). This implies that for any $h \in (-a,0)$,
\begin{equation}\label{eq:left-derivative}
\begin{aligned}
        \int_{L(a+h)}^{\bar x} s(x)\,f(x\mid a+h)\,dx
        \;-\;
        \int_{L(a)}^{\bar x}   s(x)\,f(x\mid a)\,dx
      \;-\;  c(a+h) + c(a) \leq 0.
\end{aligned}
\end{equation}
Note that the first term in the above inequality
\[ 
\begin{aligned}
\int_{L(a+h)}^{\bar x} s(x)\,f(x\mid a+h)\,dx &= \int_{L(a)}^{\bar x} s(x)\,f(x\mid a+h)\,dx + \int_{L(a+h)}^{L(a)} s(x)\,f(x\mid a+h)\,dx
\\& \geq \int_{L(a)}^{\bar x} s(x)\,f(x\mid a+h)\,dx
\end{aligned}
\] because $h<0$ and $s(x)\geq 0, f(x\mid a+h)\geq 0, \forall x$. \\
Hence, by inequality \eqref{eq:left-derivative}, we have 
\[        \int_{L(a)}^{\bar x} s(x)\,f(x\mid a+h)\,dx
        \;-\;
        \int_{L(a)}^{\bar x}   s(x)\,f(x\mid a)\,dx
      \;-\;  c(a+h) + c(a) \leq 0,\]
which is equivalent to
\[   \frac{1}{h} \Big[ \int_{L(a)}^{\bar x} s(x)\,f(x\mid a+h)\,dx
        \;-\;
        \int_{L(a)}^{\bar x}   s(x)\,f(x\mid a)\,dx \Big]
      \;-\;  \frac{c(a+h) - c(a)}{h} \geq 0, \ \ \ \forall h \in (-a,0), \]
which is equivalent to 
\[    \int_{L(a)}^{\bar x} s(x)\,\Big[ \frac{f(x\mid a+h)-f(x\mid a)}{h} \Big]\,dx
      \;-\;  \frac{c(a+h) - c(a)}{h} \geq 0, \ \ \ \forall h \in (-a,0). \]
Since the (IR) constraint implies that $\mathbb{E}[s(x)\mid a ] = c(a)+u<\infty$, Assumption \ref{ass:interchange}-(iv) allows us to take limits inside the integral. Hence, let $h \uparrow 0$ and we obtain
\[\int_{L(a)}^{\bar x} s(x) f_a(x \mid a) \,dx  -  c^\prime(a) \geq 0. \]
Combining that 
\begin{equation} \label{A2}
   E^A(a,s)  =  \int_{L(a)}^{\bar x} s(x) f(x \mid a) \,dx -c(a) = u \iff  \int_{L(a)}^{\bar x} s(x) f(x \mid a) \,dx = c(a)+u.
\end{equation}
we obtain
\begin{equation*}
  \int_{L(a)}^{\bar x} s(x) f_a(x \mid a) \,dx   \geq  \frac{c^\prime(a)}{c(a)+u} \cdot \int_{L(a)}^{\bar x} s(x) f(x \mid a) \,dx,
\end{equation*}
which implies
\begin{equation}\label{lemmaeq}
  \int_{L(a)}^{\bar x} s(x) f(x \mid a) \left( \frac{f_a(x \mid a)}{f(x \mid a)} - \frac{c^\prime(a)}{c(a)+u} \right) \,dx   \geq  0.
\end{equation}
%Suppose $\frac{f_a(x \mid a)}{f(x \mid a)}$ takes the maximum at $x^* \in [L(a), \infty]$. Then $\frac{f_a(x \mid a)}{f(x \mid a)} - \frac{c^\prime(a)}{c(a)+u} \leq  \frac{f_a(x^* \mid a)}{f(x^* \mid a)} - \frac{c^\prime(a)}{c(a)+u}, \forall x \in [L(a), \infty]$. In the following, we proceed by contradiction. \\[6pt]&= \int_{L(a)}^{\bar x} s(x)\,f_a(x\mid a)\,dx  \;-\;c'(a)\;+\;\lim_{h\uparrow  0}\frac{1}{h} \int_{L(a+h)}^{L(a)} s(x)\,f(x\mid a+h)\,dx . \\[6pt]

Suppose condition \ref{nececondition} does not hold. Then we must have $\frac{f_a(x \mid a)}{f(x \mid a)} - \frac{c^\prime(a)}{c(a)+u} < 0, \forall x \in (L(a), \, \bar x]$. Since $u_0\geq 0$, $c(a)>c(0)\geq 0$, $\forall a>0$, we have $c(a)+u > 0$, which implies that (by equation~\eqref{A2}) $\int_{L(a)}^{\bar x} s(x) f(x \mid a) \,dx > 0$. 
This further implies that set  $\{\,x\in (L(a),\, \bar x] \mid s(x)\,f(x\mid a)>0\}$ has a positive measure. Combining that $\frac{f_a(x \mid a)}{f(x \mid a)} - \frac{c^\prime(a)}{c(a)+u} < 0, \forall x \in (L(a), \, \bar x]$, we have 
\[  \int_{L(a)}^{\bar x} s(x) f(x \mid a) \left( \frac{f_a(x \mid a)}{f(x \mid a)} - \frac{c^\prime(a)}{c(a)+u} \right) \,dx   <  0,\]
which contradicts inequality \eqref{lemmaeq}. Therefore, condition \ref{nececondition} must hold. 

\textbf{(ii)} Next, we show that a solution to program~\eqref{RPJ} always exists. Let \((\hat{a}, \hat{u})\) (for example, \((0, u_0)\)) be any element in \(\mathcal{J}\), and denote \(\hat{M} = \mathbb{E}[X \mid \hat{a}] - c(\hat{a}) - \hat{u}\). By Assumption~\ref{ass:interchange}-(ii), \(\lim_{a\to +\infty} \mathbb{E}[X \mid a] - c(a) = -\infty\). Thus, there exists \(\bar{a} > \hat{a}\) such that for any \(a > \bar{a}\), \(\mathbb{E}[X \mid a] - c(a) < \hat{M} + u_0\). Consequently, for any \(u \geq u_0\) and \(a > \bar{a}\), the objective at \((a, u)\) satisfies \(\mathbb{E}[X \mid a] - c(a) - u \leq \mathbb{E}[X \mid a] - c(a) - u_0 < \hat{M}\). 
Therefore, program~\eqref{RPJ} has an optimal solution if and only if there exists an $ (a, u) \in \mathcal{J} \cap \{(a, u) \mid a \leq \bar{a}\}$ that maximizes the objective $\mathbb{E}[X(a) \mid a] - c(a) - u$.
Denote $G(a) = \sup_{x \in (L(a), \bar x]} \left( \frac{f_a(x\mid a)}{f(x\mid a)} \right)$. By Lemma \ref{suppositive}, we have $G(a) > 0$ and we can rewrite 
$\mathcal{J}$ as 
\[
\mathcal{J}
=
\left\{(a,u)\;\middle|\;  a >0,\; u \ge u_0,\;
u
\ge \frac{c'(a)}{G(a)} - c(a)
\right\}
\;\cup\;
\Big\{(0,u)\mid u\ge u_0 \Big\},
\]  
and
\[
\mathcal{J} \cap \{(a, u) \mid a \leq \bar{a}\}
=
\left\{(a,u)\;\middle|\;  a \in (0, \bar a],\; u \ge \max\{u_0,\; \frac{c'(a)}{G(a)} - c(a)\}
\right\}
\;\cup\;
\Big\{(0,u)\mid u\ge u_0 \Big\}.
\]  
Since the objective is strictly decreasing in $u$, maximizing  $\mathbb{E}[X(a) \mid a] - c(a) - u$ over $\mathcal{J} \cap \{(a, u) \mid a \leq \bar{a}\}$ is equivalent to maximizing it over 
\[  \mathcal{T}  =  \left\{(a,u)\;\middle|\;  a \in (0, \bar a],\; u = \max\{u_0,\; \frac{c'(a)}{G(a)} - c(a)\}
\right\}
\;\cup\;
\Big\{(0,u_0) \Big\}, \]
which is equivalent to the following program
\begin{equation}\label{obj}
    \max_{a \in [0, \bar a]} \mathbb{E}[X(a) \mid a] - c(a) - \max\Big\{u_0,\; \frac{c'(a)}{G(a)} - c(a)\Big\} \cdot \mathbf{1}_{\{a>0\}} -u_0 \cdot \mathbf{1}_{\{a=0\}}.
\end{equation}
By Assumption \ref{ass:interchange}-(i), $\mathbb{E}[X(a) \mid a] - c(a)$ is continuous in $a$ over $[0, \bar a]$. Also, by Assumption \ref{ass:interchange}-(iii), $G(a)$ is continuous over $(0, \bar a]$, which implies that $-\max\{u_0,\; \frac{c'(a)}{G(a)} - c(a)\} \cdot \mathbf{1}_{\{a>0\}} -u_0 \cdot \mathbf{1}_{\{a=0\}}$ is continuous in $a$ over $(0, \bar a]$. Since $ -\max\{u_0,\; \frac{c'(a)}{G(a)} - c(a)\} \leq -u_0 $, we conclude that $-\max\{u_0,\; \frac{c'(a)}{G(a)} - c(a)\} \cdot \mathbf{1}_{\{a>0\}} -u_0 \cdot \mathbf{1}_{\{a=0\}}$ is upper semicontinuous on $[0, \bar a]$ and hence the objective of program \eqref{obj} is upper semicontinuous on $[0, \bar a]$.
Since program \eqref{obj} maximizes an upper semicontinuous function over a compact set, by the upper semicontinuous Weierstrass theorem, the optimal solution exists. Hence, based on the above analysis, we conclude that there exists a $(a,u)\in \mathcal{J}$ that maximizes $\mathbb{E}[X(a) \mid a] - c(a) - u$.

\textbf{(iii)} Let \((a^{\mathcal{J}}, u^{\mathcal{J}})\) denote a maximizer of~\eqref{RPJ}. Suppose a contract \(s(x)\) implements \((a^{\mathcal{J}}, u^{\mathcal{J}})\). By definition, \((a^{\mathcal{J}}, u^{\mathcal{J}}) \in \mathcal{I}\). Since \(\mathcal{I} \subseteq \mathcal{J}\), the fact that \((a^{\mathcal{J}}, u^{\mathcal{J}})\) is a maximizer of~\eqref{RPJ} implies that it is also a maximizer of~\eqref{true}, and thus constitutes an optimal implementable pair. Therefore, since \(s(x)\) implements the optimal implementable pair, it is optimal for the \textbf{OP}.
\hfill $\blacksquare$ 
\end{proof} 
~

\begin{proof}{\textbf{Proof of Lemma \protect\ref{optimalqb}}}
~

Suppose \( s_q = b\, \mathbf{1}_{\{x \ge q\}} \) implements \((a, u)\). Given \(q\) and \(b\), the agent's expected utility is \(b(1-F(q\mid a)) - c(a)\). If \(q < L(a)\), then the agent can improve the expected utility by choosing \(a' = L^{-1}(q) < a\) to obtain \(b - c(a') >  b - c(a)\), so the (IC) constraint is violated. Therefore, we must have \(q \geq L(a)\).

If \(q > L(a)\), then the agent’s expected utility \(b(1 - F(q\mid a)) - c(a)\) is differentiable in \(a\). Optimality of \(a\) therefore requires the stationary condition \(-b\,F_a(q\mid a) - c'(a) = 0\). Combining this with the (IR) constraint \(b(1 - F(q\mid a)) = c(a) + u\) yields \(\frac{-F_a(q\mid a)}{1 - F(q\mid a)} = \frac{c'(a)}{c(a) + u}\). Therefore, the quota \(q\) must satisfy either \(q = L(a)\), or \(q > L(a)\) and $\frac{-F_a(q \mid a)}{1 - F(q \mid a)} = \frac{c'(a)}{c(a) + u}$, and the (IR) implies that \(b = \frac{c(a) + u}{1 - F(q \mid a)}\). 
\hfill $\blacksquare$ 
\end{proof}
~

\begin{proof}{\textbf{Proof of Theorem \protect\ref{qboptimalcondition}}}
~

Suppose condition \textbf{(C1)} holds. 
Under the quota-bonus contract
\(
s_{1}(x)
= \bigl(c(a^{\mathcal J}) + u^{\mathcal J}\bigr)\,\mathbf{1}_{\{x \ge L(a^{\mathcal J})\}},
\)
the agent’s expected utility is
\[
\mathbb{E}\bigl[s_{1}(X(a))\bigr] - c(a)
= 
\begin{cases}
c(a^{\mathcal J}) + u^{\mathcal J} - c(a), 
& a \ge a^{\mathcal J},\\[6pt]
\bigl(c(a^{\mathcal J}) + u^{\mathcal J}\bigr)\bigl[1 - F\bigl(L(a^{\mathcal J})\mid a\bigr)\bigr] - c(a),
& a < a^{\mathcal J}.
\end{cases}
\]
Hence, for \(a\ge a^{\mathcal J}\) the agent’s expected utility $c(a^{\mathcal J})+u^{\mathcal J}-c(a)$
is strictly decreasing in \(a\), and for \(a<a^{\mathcal J}\) condition~\textbf{(C1)} gives
\(
(c(a^{\mathcal J})+u^{\mathcal J})\bigl[1-F(L(a^{\mathcal J})\mid a)\bigr]-c(a)
\le u^{\mathcal J}.
\)
In both cases the unique best response is \(a=a^{\mathcal J}\), yielding utility \(u^{\mathcal J}\).  Hence \(s_1\) implements \((a^{\mathcal J},u^{\mathcal J})\), and by Theorem~\ref{nececondimp} there exists a quota-bonus contract that optimally solves the \textbf{OP}.

Suppose condition \textbf{(C2)} holds. Then there exists \(q>L(a^{\mathcal J})\) satisfying 
\(
\frac{-F_a(q\mid a^{\mathcal J})}{1-F(q\mid a^{\mathcal J})}
=
\frac{c'(a^{\mathcal J})}{c(a^{\mathcal J})+u}
\)
such that 
\(
\frac{1-F(q\mid a)}{1-F\!(q\mid a^{\mathcal J})}
\le
\frac{c(a)+u^{\mathcal J}}{c(a^{\mathcal J})+u^{\mathcal J}},
\forall a.
\)
Consider the quota-bonus contract $s_{2}(x)
= \frac{c(a^{\mathcal J}) + u^{\mathcal J}}{1-F\!(q\mid a^{\mathcal J}) }\,\mathbf{1}_{\{x \ge q\}}$. The agent's expected utility under $s_2(x)$ is then given by 
\[
\mathbb{E}\bigl[s_{2}(X(a))\bigr] - c(a)
= \frac{c(a^{\mathcal J}) + u^{\mathcal J}}{1-F\!(q\mid a^{\mathcal J}) }\,(1-F(q\mid a)) - c(a) \leq u^{\mathcal J}.
\]
The inequality binds when the agent chooses $a = a^{\mathcal J}$. Therefore, $s_2$ implements \((a^{\mathcal J},u^{\mathcal J})\), and by Theorem~\ref{nececondimp} there exists a quota-bonus contract that optimally solves the \textbf{OP}.
\hfill $\blacksquare$  
\end{proof}

\begin{proof}{\textbf{Proof of Theorem \protect\ref{weakerthanFOAVQB}}}
~

Let $s_r = b \cdot \mathbf{1}_{x\geq q}$ be a quota-bonus contract that solves the \textbf{RP} and also solves the \textbf{OP}. Let $a_r$ be the agent's optimal effort level under $s_r$, and $u_r = E^A(a_r, s_r)$ be the resulting agent's expected utility. By definition, $s_r$ implements $(a_r, u_r)$. 
Also, \(s_r,a_r\) must satisfy RIC and IR, which is given by 
\[ -bF_a(q\mid a_r) = c^\prime(a_r) \quad \text{ and }  \quad   b (1-F(q\mid a_r)) = c(a_r) + u_r. \] 
It follows that $q$ must satisfy $q> L(a_r)$. Otherwise, if $q \leq L(a_r)$, then $\frac{\partial^{+}}{\partial a} F(q\mid a_r) = 0$, which would violate the first equality above (i.e., RIC) since $c^\prime(a_r)>0$. Combining RIC and IR, we obtain $\frac{-F_a(q\mid a_r)}{1-F(q\mid a_r)}  = \frac{c^\prime(a_r)}{c(a_r) + u_r}$, and $b = \frac{c(a_r) + u_r}{1-F(q\mid a_r)}$. Since $s_r$ also solves the \textbf{OP}, it satisfies IC, i.e., \(
b (1-F(q\mid a)) - c(a)
\le
b (1-F(q\mid a_r)) - c(a_r) = u_r,
\ \ \forall a\geq 0
\), which is equivalent to    \(
\frac{1-F(q\mid a)}{1-F\!(q\mid a_r)}
\le
\frac{c(a)+u_r}{c(a_r)+u_r},
\ \ \forall a\geq 0.
\)
Now, if $a_r$ is a solution to program \eqref{RPJ}, then by definition, condition \textbf{(C2)} holds. Hence, it remains to show that $(a_r, u_r)$ is a solution to program \eqref{RPJ}. 

We proceed by contradiction. Suppose instead that $(a_r, u_r)$ is not a solution to program \eqref{RPJ}. Then to show the contradiction, it is sufficient to show that either $s_r$ does not optimally solve \textbf{RP} or $s_r$ does not optimally solve \textbf{OP}. To show this, we will construct another contract that yields a higher principal's expected payoff than $s_r$ yields in the \textbf{RP} or the \textbf{OP}.

Recall that program \eqref{RPJ} is $\max_{(a, u) \in \mathcal{J}} \left\{ \mathbb{E}[X(a) \mid a] - c(a) - u \right\}$, where $\mathcal J$ is given in theorem \ref{nececondimp} and the objective is exactly the principal's expected payoff given $(a, u)$. Let $(a^\mathcal{J}, u^\mathcal{J})$ be a solution to it (theorem \ref{nececondimp} guarantees the existence). 

\textbf{Case 1.} If $a^\mathcal{J} = 0$, then $(0, u^\mathcal{J})$ can be implemented by the contract $s_0 \equiv u^\mathcal{J}$. This implies that $s_r$ does not optimally solve \textbf{OP}: the principal's expected utility under $s_0$ is strictly higher than that under $s_r$ since $(a_r, u_r)$ is not a solution to program \eqref{RPJ}.

\textbf{Case 2.} If $a^\mathcal{J} > 0$, then since $(a_r, u_r)$ is not a solution to program \eqref{RPJ}, there exists $\epsilon>0$ such that $\mathbb{E}[X(a^\mathcal{J}) \mid a^\mathcal{J}] - c(a^\mathcal{J}) - (u^\mathcal{J} + \epsilon) > \mathbb{E}[X(a_r) \mid a_r] - c(a_r) - u_r$. That  $(a^\mathcal{J}, u^\mathcal{J}) \in \mathcal{J}$ implies that $(a^\mathcal{J}, u^\mathcal{J}+ \epsilon) \in \mathcal{J}$ because $u^\mathcal{J}+ \epsilon > u^\mathcal{J} \geq u_0$ and 
\[\sup_{x \in (L(a^\mathcal{J}),\, \bar{x}]} \frac{f_a(x \mid a^\mathcal{J})}{f(x \mid a^\mathcal{J})} \geq \frac{c'(a^\mathcal{J})}{c(a^\mathcal{J}) + u^\mathcal{J}} > \frac{c'(a^\mathcal{J})}{c(a^\mathcal{J}) + u^\mathcal{J} + \epsilon}. \]
We now show $s_r$ does not optimally solve \textbf{RP}. Since $(a^\mathcal{J}, u^\mathcal{J}+ \epsilon)$ gives a higher principal's expected payoff than $(a_r, u_r)$, it remains to show that there exists a contract $s$ such that $s $ and $ a^\mathcal{J}$ satisfy RIC and $E^A(s, a^\mathcal{J}) = u^\mathcal{J} + \epsilon$. 

%\textbf{Subcase-(i): Fixed Support.} We first focus on the fixed-support settings, where $L(a)$ is a constant over $a>0$, denoted by $\underline x$. We claim that there must exist some $\hat{x} \in (\underline x,\, \bar{x})$ such that $\frac{f_a(\hat{x} \mid a^\mathcal{J})}{f(\hat{x} \mid a^\mathcal{J})} = \frac{c'(a^\mathcal{J})}{c(a^\mathcal{J}) + u^\mathcal{J} + \epsilon}$. To show the existence of such $\hat{x}$, we only need to show there exists some $x' \in (\underline x,\, \bar{x}]$ such that $\frac{f_a(x' \mid a^\mathcal{J})}{f(x' \mid a^\mathcal{J})} <  \frac{c'(a^\mathcal{J})}{c(a^\mathcal{J}) + u^\mathcal{J} + \epsilon}$ because $\sup_{x \in (\underline x,\, \bar{x}]} \frac{f_a(x \mid a^\mathcal{J})}{f(x \mid a^\mathcal{J})} > \frac{c'(a^\mathcal{J})}{c(a^\mathcal{J}) + u^\mathcal{J} + \epsilon}$ and $\frac{f_a(x \mid a^\mathcal{J})}{f(x \mid a^\mathcal{J})}$ is   continuous in $x$. Such $x'$ must exist because otherwise, $\frac{f_a(x \mid a^\mathcal{J})}{f(x \mid a^\mathcal{J})} \geq  \frac{c'(a^\mathcal{J})}{c(a^\mathcal{J}) + u^\mathcal{J} + \epsilon} >0$ for any $x\in (\underline x, \bar x]$, which means $\int_{\underline{x}}^{\bar{x}} f_a(x \mid a^\mathcal{J})\, dx > 0$, which contradicts that \(\int_{\underline{x}}^{\bar{x}} f(x \mid a^\mathcal{J})\, dx = 1 \implies \int_{\underline{x}}^{\bar{x}} f_a(x \mid a^\mathcal{J})\, dx = 0   \).

\textbf{Case 2-(i).} Suppose there exists some $x' \in (L(a^\mathcal{J}),\, \bar{x}]$ such that $\frac{f_a(x' \mid a^\mathcal{J})}{f(x' \mid a^\mathcal{J})} <  \frac{c'(a^\mathcal{J})}{c(a^\mathcal{J}) + u^\mathcal{J} + \epsilon}$.
Because $\sup_{x \in (L(a^\mathcal{J}),\, \bar{x}]} \frac{f_a(x \mid a^\mathcal{J})}{f(x \mid a^\mathcal{J})} > \frac{c'(a^\mathcal{J})}{c(a^\mathcal{J}) + u^\mathcal{J} + \epsilon}$ and $\frac{f_a(x \mid a^\mathcal{J})}{f(x \mid a^\mathcal{J})}$ is  continuous in $x$ on $(L(a^\mathcal{J}),\, \bar{x}]$, there must exist some $\hat{x} \in (L(a^\mathcal{J}),\, \bar{x})$ such that $\frac{f_a(\hat{x} \mid a^\mathcal{J})}{f(\hat{x} \mid a^\mathcal{J})} = \frac{c'(a^\mathcal{J})}{c(a^\mathcal{J}) + u^\mathcal{J} + \epsilon}$.

Now we use the existence of such $\hat{x}$ to construct a contract $s$ such that $s$ and $ a^\mathcal{J}$ satisfy RIC and $E^A(s, a^\mathcal{J}) = u^\mathcal{J} + \epsilon$. Since $\frac{f_a(\hat{x} \mid a^\mathcal{J})}{f(\hat{x} \mid a^\mathcal{J})} = \frac{c'(a^\mathcal{J})}{c(a^\mathcal{J}) + u^\mathcal{J} + \epsilon} < \sup_{x \in (\underline x,\, \bar{x}]} \frac{f_a(x \mid a^\mathcal{J})}{f(x \mid a^\mathcal{J})}$ and $\frac{f_a(x \mid a^\mathcal{J})}{f(x \mid a^\mathcal{J})}$ is continuous in $x$, there exists a constant $M_1 < \sup_{x \in (\underline x,\, \bar{x}]} \frac{f_a(x \mid a^\mathcal{J})}{f(x \mid a^\mathcal{J})}$ such that set \[
A_1 
= \Biggl\{\,
x \;\Bigg|\;  M_1> 
\frac{f_a(x \mid a^{\mathcal{J}})}{f(x \mid a^{\mathcal{J}})} 
> 
\frac{f_a(\hat{x} \mid a^{\mathcal{J}})}{f(\hat{x} \mid a^{\mathcal{J}})} 
\Biggr\} 
\;\cap\; (\underline{x}, \bar{x}] 
\]
is measurable and has strictly positive measure (Lebesgue). Similarly, since $\frac{f_a(\hat{x} \mid a^\mathcal{J})}{f(\hat{x} \mid a^\mathcal{J})}> \frac{f_a(x' \mid a^\mathcal{J})}{f(x' \mid a^\mathcal{J})}$ for some $x' \in (\underline{x}, \bar{x}]$, by continuity, there exists a constant $M_2 >\frac{f_a(x' \mid a^\mathcal{J})}{f(x' \mid a^\mathcal{J})} $ such that set 
\[
A_2 
= \Biggl\{\, 
x \;\Bigg|\; M_2< 
\frac{f_a(x \mid a^{\mathcal{J}})}{f(x \mid a^{\mathcal{J}})} 
< 
\frac{f_a(\hat{x} \mid a^{\mathcal{J}})}{f(\hat{x} \mid a^{\mathcal{J}})} 
\Biggr\} 
\;\cap\; (\underline{x}, \bar{x}] 
\]
is measurable and has strictly positive measure (Lebesgue). We construct the contract $s$ by letting
\[s(x) \;=\;
\begin{cases}
b_1, & x \in  A_1,\\
b_2, & x \in  A_2.
\end{cases}\] 
We now show there exists $b_1, b_2>0$ such that $s$ and $a^\mathcal{J}$ satisfy RIC and $E^A(s, a^\mathcal{J}) = u^\mathcal{J} + \epsilon$. Equivalently, we need to find $b_1, b_2>0$ such that 
\[
b_1 \int_{A_1} f_a(x \mid a^{\mathcal J})\,\mathrm{d}x
\;+\;
b_2 \int_{A_2} f_a(x \mid a^{\mathcal J})\,\mathrm{d}x
\;-\;
c'(a^{\mathcal J}) \;=\; 0.
\]
\[
b_1 \int_{A_1} f(x \mid a^{\mathcal J})\,\mathrm{d}x
\;+\;
b_2 \int_{A_2} f(x \mid a^{\mathcal J})\,\mathrm{d}x
\;-\;
c(a^{\mathcal J}) \;=\; u^{\mathcal J} + \epsilon.
\]

For simplicity, we denote
\(
I_1:=\int_{A_1} f_a(x \mid a^{\mathcal J})\,\mathrm{d}x \),
\(I_2:=\int_{A_2} f_a(x \mid a^{\mathcal J})\,\mathrm{d}x \),
\(J_1:=\int_{A_1} f(x \mid a^{\mathcal J})\,\mathrm{d}x \),
\(J_2:=\int_{A_2} f(x \mid a^{\mathcal J})\,\mathrm{d}x
\), and 
\(
\Delta := I_1 J_2 - I_2 J_1 
\).
The above Linear system gives 
\[
b_1
= \frac{ c'(a^{\mathcal J})\, J_2 - \bigl(c(a^{\mathcal J}) + u^{\mathcal J} + \epsilon\bigr)\, I_2 }{\Delta},
\qquad
b_2
= \frac{ \bigl(c(a^{\mathcal J}) + u^{\mathcal J} + \epsilon\bigr)\, I_1 - c'(a^{\mathcal J})\, J_1 }{\Delta}.
\]
Note that by the definition of $A_1, A_2$ and $\hat{x}$, we have $0<J_1\leq 1$ and $0<J_2\leq 1$ because both $A_1$ and $A_2$ have strictly positive measure and $f$ is a density function; also $|I_1|<\infty$ and $|I_2|<\infty$ because $|f_a(x\mid a^{\mathcal J})|$ is bounded by integrable functions in both $A_1$ and $A_2$.
Hence, it remains to show that $\Delta \neq 0$ and $b_1, b_2>0$. Note that
\[I_1 = \int_{A_1} \frac{f_a(x \mid a^{\mathcal J})}{f(x \mid a^{\mathcal J})} \cdot f(x \mid a^{\mathcal J}) \,\mathrm{d}x >\frac{f_a(\hat{x} \mid a^{\mathcal{J}})}{f(\hat{x} \mid a^{\mathcal{J}})}  \cdot \int_{A_1}  f(x \mid a^{\mathcal J})\,\mathrm{d}x = \frac{c'(a^{\mathcal J})}{c(a^{\mathcal J}) + u^{\mathcal J} + \epsilon} \cdot J_1 , \]
\[I_2 = \int_{A_2} \frac{f_a(x \mid a^{\mathcal J})}{f(x \mid a^{\mathcal J})} \cdot f(x \mid a^{\mathcal J}) \,\mathrm{d}x < \frac{f_a(\hat{x} \mid a^{\mathcal{J}})}{f(\hat{x} \mid a^{\mathcal{J}})}  \cdot \int_{A_2}  f(x \mid a^{\mathcal J})\,\mathrm{d}x = \frac{c'(a^{\mathcal J})}{c(a^{\mathcal J}) + u^{\mathcal J} + \epsilon} \cdot J_2.  \]
The above inequalities imply that $c'(a^{\mathcal J})\, J_2 - \bigl(c(a^{\mathcal J}) + u^{\mathcal J} + \epsilon\bigr)\, I_2 > 0$ and $ \bigl(c(a^{\mathcal J}) + u^{\mathcal J} + \epsilon\bigr)\, I_1 - c'(a^{\mathcal J})\, J_1$, and $\Delta > 0$ because 
\[I_1J_2 > \frac{c'(a^{\mathcal J})}{c(a^{\mathcal J}) + u^{\mathcal J} + \epsilon} \cdot J_1 J_2 > I_2J_1 .\]
This further implies that $\Delta \neq 0$, $b_1>0$ and $b_2>0$. This completes the proof by contradiction.

\textbf{Case 2-(ii).} Suppose there does not exists an $x' \in (L(a^\mathcal{J}),\, \bar{x}]$ such that $\frac{f_a(x' \mid a^\mathcal{J})}{f(x' \mid a^\mathcal{J})} <  \frac{c'(a^\mathcal{J})}{c(a^\mathcal{J}) + u^\mathcal{J} + \epsilon}$. Then for any $x \in  (L(a^\mathcal{J}),\, \bar{x}]$, we have $\frac{f_a(x \mid a^\mathcal{J})}{f(x \mid a^\mathcal{J})} \geq  \frac{c'(a^\mathcal{J})}{c(a^\mathcal{J}) + u^\mathcal{J} + \epsilon} >0$, which implies that $f_a(x \mid a^\mathcal{J})>0, \forall x \in  (L(a^\mathcal{J}),\, \bar{x}]$. We claim that we must have $f(L(a^{\mathcal J})\mid a^{\mathcal J})\,L'(a^{\mathcal J})>0$. Suppose to the contrary that \(f(L(a^{\mathcal J})\mid a^{\mathcal J})\,L'(a^{\mathcal J})= 0\). Then \(\int_{L(a)}^{\bar x} f(x\mid a)\,dx=1\) for all \(a>0\) implies that \(\int_{L(a^{\mathcal J})}^{\bar x} f_a(x\mid a^{\mathcal J})\,dx = f(L(a^{\mathcal J})\mid a^{\mathcal J})\,L'(a^{\mathcal J}) = 0\). This contradicts $f_a(x \mid a^\mathcal{J})>0, \forall x \in  (L(a^\mathcal{J}),\, \bar{x}]$.

Now we use $f(L(a^{\mathcal J})\mid a^{\mathcal J})\,L'(a^{\mathcal J})>0$ to show there exists a contract $s$ such that $s$ and $ a^\mathcal{J}$ satisfy RIC and $E^A(s, a^\mathcal{J}) = u^\mathcal{J} + \epsilon$. First, because $L$ is continuously differentiable on $(0,\infty)$, we can find $a_2> a_1>0$ such that $a^\mathcal{J} \in (a_1, a_2)$ and $L'(a) >0$ for any $a \in (a_1, a_2)$. Hence, $L(a_1)<L(a^\mathcal{J})<L(a_2)$. 
We construct the contract $s$ as follows
\[s(x) \;=\;
\begin{cases}
b_1, & x \in  [L(a_1), L(a_2)),\\
b_2, & x \in  [L(a_2), +\infty).
\end{cases}\] 
We now show there exists $b_1, b_2>0$ such that $s$ and $a^\mathcal{J}$ satisfy RIC and $E^A(s, a^\mathcal{J}) = u^\mathcal{J} + \epsilon$. Equivalently, we need to find $b_1, b_2>0$ such that 
\[
E^A(a^{\mathcal J}, s) =  b_1 \int_{L(a^{\mathcal J})}^{L(a_2)} f(x \mid a^{\mathcal J})\,\mathrm{d}x
\;+\;
b_2 \int_{L(a_2)}^{\infty} f(x \mid a^{\mathcal J})\,\mathrm{d}x
\;-\;
c(a^{\mathcal J}) \;=\; u^{\mathcal J} + \epsilon,
\]
and $\frac{\partial }{\partial a}E^A(a^{\mathcal J}, s) = 0$, which is equivalent to
\[ 
b_1 \int_{L(a^{\mathcal J})}^{L(a_2)} f_a(x \mid a^{\mathcal J})\,\mathrm{d}x
\;+\;
b_2 \int_{L(a_2)}^{\infty} f_a(x \mid a^{\mathcal J})\,\mathrm{d}x
\;-\;  b_1 f(L(a^{\mathcal J})\mid a^{\mathcal J})\,L'(a^{\mathcal J})
-
c'(a^{\mathcal J}) = 0.
\]
For simplicity, we denote \(I_1 := \int_{L(a^{\mathcal J})}^{L(a_2)} f_a(x \mid a^{\mathcal J})\,dx\), \(I_2 := \int_{L(a_2)}^{\infty} f_a(x \mid a^{\mathcal J})\,dx\), \(J_1 := \int_{L(a^{\mathcal J})}^{L(a_2)} f(x \mid a^{\mathcal J})\,dx\), \(J_2 := \int_{L(a_2)}^{\infty} f(x \mid a^{\mathcal J})\,dx\), and \(H := f(L(a^{\mathcal J}) \mid a^{\mathcal J})\,L'(a^{\mathcal J})\).
We thus obtain the linear system \[J_1 b_1 + J_2 b_2 = c(a^{\mathcal J}) + u^{\mathcal J} + \epsilon,\] \[(I_1 - H) b_1 + I_2 b_2 = c'(a^{\mathcal J}).\]
Solving this system yields
\[b_1 = \dfrac{(c(a^{\mathcal J})+u^{\mathcal J}+\epsilon)\,I_2 - c'(a^{\mathcal J})\,J_2}{\,J_1 I_2 - J_2(I_1 - H)\,}, \]
\[b_2 = \dfrac{J_1 c'(a^{\mathcal J}) - (I_1 - H)\,(c(a^{\mathcal J})+u^{\mathcal J}+\epsilon)}{\,J_1 I_2 - J_2(I_1 - H)\,}.\]
It remains to show that $b_1, b_2>0$, $I_1,I_2<+\infty$, and $J_1 I_2 - J_2(I_1 - H) \neq 0$. Note that we have already shown $f_a(x \mid a^\mathcal{J})>0, \forall x \in  (L(a^\mathcal{J}),\, \bar{x}]$, which means $I_1, I_2>0$, and by definition $J_1, J_2, H >0$. Also, since $\frac{f_a(x \mid a^\mathcal{J})}{f(x \mid a^\mathcal{J})} \geq   \frac{c'(a^\mathcal{J})}{c(a^\mathcal{J}) + u^\mathcal{J} + \epsilon}, \forall x \in  (L(a^\mathcal{J}),\, \bar{x}]$, we have 
\[I_2 = \int_{L(a_2)}^{\infty} \frac{f_a(x \mid a^{\mathcal J})}{f(x \mid a^{\mathcal J})}\, f(x \mid a^{\mathcal J})\,\mathrm{d}x \;>\;  \frac{c'(a^{\mathcal J})}{c(a^{\mathcal J}) + u^{\mathcal J} + \epsilon}\, J_2.
\] This implies that the numerator of $b_1$, $(c(a^{\mathcal J})+u^{\mathcal J}+\epsilon)\,I_2 - c'(a^{\mathcal J})\,J_2>0$.
Also, \(\int_{L(a)}^{\bar x} f(x\mid a)\,dx=1\) for all \(a>0\) implies that \(\int_{L(a^{\mathcal J})}^{\bar x} f_a(x\mid a^{\mathcal J})\,dx = f(L(a^{\mathcal J})\mid a^{\mathcal J})\,L'(a^{\mathcal J})\), which further implies that
\[I_1 < \int_{L(a^{\mathcal J})}^{\bar x}  f_a(x \mid a^{\mathcal J})\,\mathrm{d}x \;=\; f(L(a^{\mathcal J}) \mid a^{\mathcal J})\,L'(a^{\mathcal J}) = H.
\]
This further implies that the numerator of $b_2$, $J_1 c'(a^{\mathcal J}) - (I_1 - H)\,(c(a^{\mathcal J})+u^{\mathcal J}+\epsilon)>0$, and the denominator of both $b_1$ and $b_2$, $J_1 I_2 - J_2(I_1 - H)>0$. By the same argument that establishes $I_1<H$, we also conclude $I_2<H$. Combined, we have demonstrated that $b_1, b_2>0$, $I_1,I_2<+\infty$, and $J_1 I_2 - J_2(I_1 - H) \neq 0$. This completes our proof of Case 2-(ii).

Since Case 2-(i) and Case 2-(ii) are complementary, we have completed the proof of Case 2. Together with Case 1 and the previous arguments, this completes the proof.
\hfill $\blacksquare$  
\end{proof}
\medskip

\end{APPENDICES}

\begin{figure}[htbp]
  \FIGURE
    {\includegraphics[width=0.6\textwidth]{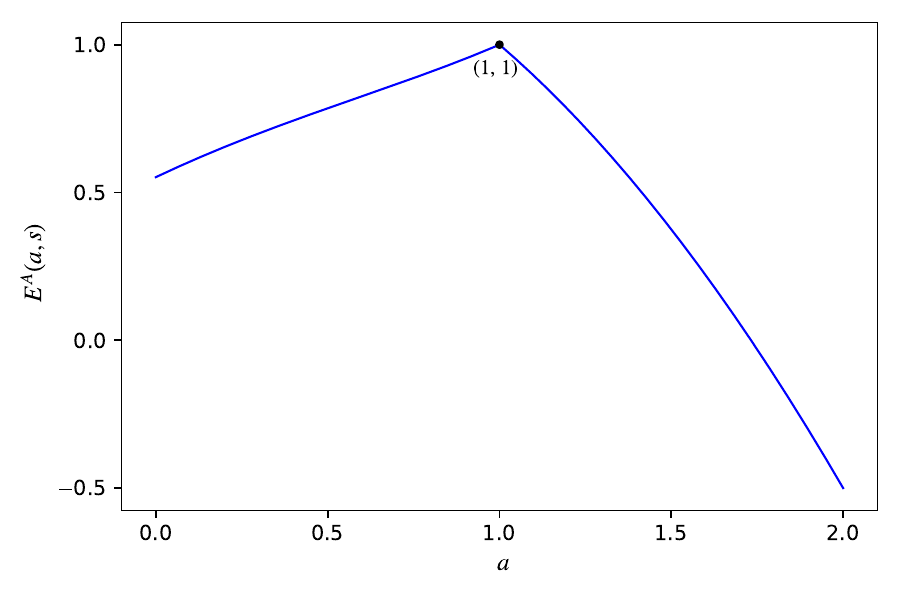}}
    {Agent’s expected utility as a function of effort under contract $\hat{s}(x)$. \label{fig:agent-utility}}
    {The unique maximizer \(a = 1\) is a kink point: the left derivative of \(E^A(a, \hat{s})\) at \(a = 1\) is \(\tfrac{1}{2}\), while the right derivative is \(-1\).
}
\end{figure}

% Appendix

\end{document}